\documentclass[11pt]{article}

\usepackage{indentfirst}
\usepackage{amsmath}
\usepackage{amsthm}
\usepackage{color}
\usepackage{eufrak}
\usepackage{verbatim}
\usepackage[makeroom]{cancel}
\usepackage{rotating}
\usepackage{float}
\usepackage{fullpage}
\usepackage{amssymb}
\usepackage{amsthm}
\usepackage{mathrsfs}
\usepackage{mathtools}
\usepackage{epsfig}
\usepackage{graphicx}
\usepackage{indentfirst}
\usepackage{framed}
\usepackage[numbers]{natbib}
\usepackage{color}
\usepackage{hyperref}
\usepackage{graphicx}
\usepackage{hyperref}
\usepackage{float}

\usepackage{bm}
\usepackage{cleveref}
\usepackage[labelfont=bf]{caption}
\usepackage{algorithm}
\usepackage[noend]{algpseudocode}
\usepackage{soul}
\usepackage{enumitem}
\usepackage{lineno}

\definecolor{corlinks}{RGB}{0,0,170}
\definecolor{cormenu}{RGB}{0,0,170}
\definecolor{corurl}{RGB}{0,0,170}

\hypersetup{
	colorlinks=true,
	urlcolor=corlinks,
	linkcolor=corlinks,
	menucolor=cormenu,
	citecolor=corlinks,
	pdfborder= 0 0 0
}

\newtheorem{theorem}{Theorem}

\newtheorem{problem}[theorem]{Problem}

\newtheorem{corollary}[theorem]{Corollary}

\newtheorem{proposition}[theorem]{Proposition}
\newtheorem{remark}[theorem]{Remark}

\newtheorem{fact}[theorem]{Fact}

\DeclareMathOperator{\poly}{poly}
\DeclareMathOperator{\polylog}{polylog}

\DeclareMathOperator*\Exp{{\bf E}}
\DeclareMathOperator*\Prob{{\bf Pr}}
\newcommand{\bool}{\left\{0,1\right\}}
\newcommand{\Kt}{\mathsf{Kt}}
\newcommand{\rKt}{\mathsf{rKt}}

\newcommand{\RTIME}{\mathsf{RTIME}}
\newcommand{\UP}{\mathsf{UP}}
\newcommand{\BPTIME}{\mathsf{BPTIME}}
\newcommand{\Dist}{\mathsf{Dist}}
\newcommand{\Avg}{\mathsf{Avg}}
\newcommand{\BPE}{\mathsf{BPE}}

\newcommand{\K}{\mathsf{K}}
\newcommand{\rK}{\mathsf{rK}}
\newcommand{\pK}{\mathsf{pK}}

\newcommand{\eqdef}{\stackrel{\rm def}{=}}

\newcommand{\BPP}{\mathsf{BPP}}

\newcommand{\AM}{\mathsf{AM}}

\newcommand{\MA}{\mathsf{MA}}

\newcommand{\NP}{\mathsf{NP}}
\newcommand{\PH}{\mathsf{PH}}

\def\caC{\mathcal{C}}

\usepackage{algorithm}
\usepackage[noend]{algpseudocode}

\definecolor{blue-violet}{rgb}{0.54, 0.17, 0.89}
\definecolor{darkorange}{rgb}{1.0, 0.55, 0.0}

\renewcommand{\P}{\mathsf{P}}
\newcommand{\promise}{\mathsf{promise}\text{-}}

\begin{document}

	\title{Theory and Applications of Probabilistic  Kolmogorov Complexity\vspace{0.5cm}}

	
	\author{
	Zhenjian Lu\thanks{University of Warwick, UK. \texttt{E-mail:~zhen.j.lu@warwick.ac.uk}}\and
	Igor C. Oliveira\thanks{University of Warwick, UK. \texttt{E-mail:~igor.oliveira@warwick.ac.uk}}
	\vspace{0.6cm}}
	
	
	\maketitle
	
	\vspace{-0.6cm}

	\begin{abstract}
Diverse applications of  Kolmogorov complexity to learning \citep{DBLP:conf/coco/CarmosinoIKK16}, circuit complexity \citep{DBLP:conf/coco/OliveiraPS19},  cryptography \citep{DBLP:conf/focs/LiuP20}, average-case complexity \citep{DBLP:conf/stoc/Hirahara21}, and proof search \citep{Kra22} have been discovered in recent years. Since the running time of algorithms is a key resource in these fields, it is crucial in the corresponding arguments to consider \emph{time-bounded} variants of Kolmogorov complexity. While fruitful interactions between time-bounded Kolmogorov complexity and different areas of theoretical computer science have been known for quite  a while (e.g.,~\citep{DBLP:conf/stoc/Sipser83a,  DBLP:journals/siamcomp/Ko91, DBLP:journals/siamcomp/AllenderBKMR06, DBLP:conf/coco/AntunesF09}, to name a few), the aforementioned results have led to a renewed interest in this topic. 

The theory of Kolmogorov complexity is well understood, but many useful results and properties of Kolmogorov complexity are not known to hold in time-bounded settings. Unfortunately, this creates technical difficulties or leads to conditional results when applying methods from time-bounded Kolmogorov complexity to algorithms and complexity theory. Perhaps even more importantly, in many cases it is  desirable or even necessary to consider \emph{randomised} algorithms. Since random strings have high  complexity, the classical theory of time-bounded Kolmogorov complexity might be inappropriate or simply cannot be applied in such contexts.

To mitigate these issues and develop a more robust theory of time-bounded Kolmogorov complexity that survives in the important setting of randomised computations, some recent papers \citep{DBLP:conf/icalp/Oliveira19, DBLP:conf/icalp/LuO21, LOS21, GKLO22, LOZ22} have explored \emph{probabilistic} notions of time-bounded Kolmogorov complexity, such as $\rKt$ complexity \citep{DBLP:conf/icalp/Oliveira19}, $\rK^t$ complexity \citep{LOS21}, and $\pK^t$ complexity \citep{GKLO22}. These measures consider different ways of encoding an object via a \emph{probabilistic representation}. In this survey, we provide an introduction to probabilistic time-bounded Kolmogorov complexity and its applications, highlighting many open problems and  research directions.
	\end{abstract}

	\setcounter{tocdepth}{2}
	
	\newpage
	
	\tableofcontents
	
	\newpage

	\section{Introduction}\label{sec:intro}

Consider an arbitrary binary string $x \in \{0,1\}^*$, e.g.,
	\begin{equation}\label{eq:example}
	    	x = 1010101010101010101011110010110000101011\,. 
	\end{equation}
The Kolmogorov complexity of $x$, $\K(x)$, is the length $|M|$ of the shortest program $M$ that prints $x$ when computing over the empty input string.\footnote{We formally define (time-bounded) Kolmogorov complexity in \Cref{sec:basic_definitions}.} Intuitively, $\K(x)$ can be seen as a measure of the ``randomness'' of $x$, in the sense that simple strings exhibiting an apparent pattern have bounded Kolmogorov complexity (e.g., the leftmost 20 bits of the string $x$ from \Cref{eq:example}), while a typical random $n$-bit string has $\K(x)$ close to $n$, i.e., it cannot be compressed. The investigation of Kolmogorov complexity has uncovered surprising connections to distant areas of mathematics and computer
science, ranging from computability, logic, and algorithm design to number theory,   combinatorics, statistics and a number of other fields. We refer to \citep{KolmBook2, li-vit:b:kolmbook-four} for a comprehensive treatment of Kolmogorov complexity and its  applications.

	Despite the appealing nature and wide applicability of Kolmogorov complexity, its results and techniques  tend to be inappropriate in settings where the \emph{running time of algorithms is of concern}, e.g., in complexity theory, computational learning theory, and cryptography. This is because $\K(x)$ does not take into account the time that the machine $M$ takes to output $x$. To address this issue, several authors have contributed to the development of \emph{time-bounded} Kolmogorov complexity. In order to proceed with our discussion, we describe two prominent time-bounded Kolmogorov complexity notions. (A formal treatment appears in \Cref{sec:basic_definitions}.)
	
	In an influential paper, Levin \citep{DBLP:journals/iandc/Levin84} introduced $\Kt(x)$, a variant of Kolmogorov complexity that simultaneously takes into account the running time $t$ and description length $|M|$ of all programs $M$ that output $x$. More precisely, given a string $x \in \{0,1\}^*$, we let
	\begin{eqnarray}
	    \Kt(x) =  \min_{M, \;t \geq 1} \left\{|M| + \lceil \log t \rceil \mid M~\text{outputs}~x~\text{in}~t~\text{steps} \right \}.
	\end{eqnarray}
		To provide intuition and give a concrete example of the usefulness of this time-bounded variant of Kolmogorov complexity to algorithms and complexity theory, we consider the following computational problem at the intersection of mathematics and computer science:\\
	
	\vspace{-0.2cm}
	
	\noindent \emph{Explicit Construction of Primes:} Given an integer $n \geq 2$, deterministically compute an $n$-bit prime number.\footnote{For instance, the string $x$ in  \Cref{eq:example} is a $40$-bit prime ($733008047147$ in decimal representation).}\\
	
		\vspace{-0.2cm}

	The fastest known algorithm that solves this problem runs in time $\widetilde{O}(2^{n/2})$ \citep{DBLP:journals/jal/LagariasO87}, and it is a longstanding open problem to improve this bound (see~\citep{MR2869058}). Let $A(n)$ denote this procedure, and consider the sequence $\{p_n\}_{n \geq 2}$ of primes output by $A(n)$. Since we can encode the fixed algorithm $A$ using $O(1)$ bits and any fixed number $n$ using $O(\log n)$ bits, it follows that some program $M$ of description length $O(\log n)$ runs in time $t = \widetilde{O}(2^{n/2})$ and prints $p_n$. Consequently, there is an  $n$-bit prime $p_n$ such that $\Kt(p_n) \leq O(1) + O(\log n) + \log t \leq n/2 + O(\log n)$. More generally, a faster algorithm yields improved bounds on the $\Kt$ complexity of some sequence of prime numbers. Conversely, it is possible to prove that if there is a sequence $\{q_n\}_{n \geq 2}$ of $n$-bit primes such that $\Kt(q_n) = \lambda_n$, then the problem of explicit constructing primes can be solved in time $\widetilde{O}(2^{\lambda_n})$.\footnote{As discovered by Levin, this is achieved by an algorithm that attempts to compute an $n$-bit prime by carefully simulating all programs of small description length for an appropriate number of steps until an $n$-bit prime is found.} This shows that one can completely capture the problem of explicitly constructing primes via time-bounded Kolmogorov complexity!
	
	Note that in $\Kt$ complexity the time bound is not fixed and depends on the best possible description of $x$. In some contexts, it is desirable to restrict attention to programs $M$ that run under a specified time bound $t(n)$, e.g., in time $ \leq n^3$. This is captured by $\K^t$ complexity (see, e.g., \citep{DBLP:conf/stoc/Sipser83a}), where $t \colon \mathbb{N} \to \mathbb{N}$ is a fixed function:
	\begin{eqnarray}
	    \K^t(x) =  \min_{M} \left\{|M|  \,\mid\, M~\text{outputs}~x~\text{in}~t(|x|)~\text{steps} \right \}.
	\end{eqnarray}
	As a recent application of time-bounded Kolmogorov complexity, Liu and Pass \citep{DBLP:conf/focs/LiuP20} connected one-way functions (OWF), a primitive that is essential to cryptography, to the computational difficulty of estimating the $\K^t$ complexity of an input string $x$, when $t$ is a fixed polynomial. A bit more precisely, they showed that OWFs exist if and only if it is computationally hard on average to estimate $\K^t(x)$ for a random input string $x$ (see their paper for the exact statement). This provides another striking example of the power and reach of time-bounded Kolmogorov complexity.
	
	While connections between time-bounded Kolmogorov complexity and different areas of theoretical computer science have been known for a long time (see, e.g.,~\citep{DBLP:conf/stoc/Sipser83a,  DBLP:journals/siamcomp/Ko91, DBLP:journals/siamcomp/AllenderBKMR06, DBLP:conf/coco/AntunesF09}), recent applications of it to cryptography \citep{DBLP:conf/focs/LiuP20, DBLP:conf/coco/RenS21, DBLP:conf/crypto/LiuP21}, learning \citep{DBLP:conf/coco/CarmosinoIKK16, HN21}, average-case complexity \citep{DBLP:conf/stoc/Hirahara21}, circuit complexity \citep{DBLP:conf/coco/OliveiraPS19}, and proof search \citep{Kra22}  have led to much interest in this topic and to a number of related developments.  We refer the reader to these papers and to \citep{allender1992applications,  allender2001worlds, fortnow2004kolmogorov, TroyLeeThesis,  allender2017complexity, li-vit:b:kolmbook-four, allender2021vaughan} for more information on different time-bounded Kolmogorov complexity measures and their applications.\\

 \noindent \textbf{Probabilistic (Time-Bounded) Kolmogorov Complexity.} The need to use time-bounded Kolmogorov complexity in certain applications can create issues that are not present in the case of (time-unbounded) Kolmogorov complexity. More precisely, several central results from Kolmogorov complexity are not known to hold in a time-bounded setting. Some of them do survive under a plausible assumption (e.g.,~a \emph{source coding theorem} holds for $\K^t$ under a strong derandomisation assumption \citep{DBLP:conf/coco/AntunesF09}), but this leads to \emph{conditional} results only. In other cases, the validity of a result in the setting of time-bounded Kolmogorov complexity is closely tied to a longstanding open problem in complexity theory (e.g.,~the computational difficulty of estimating $\K^t(x)$ and the aforementioned connection to OWFs \citep{DBLP:conf/focs/LiuP20}). We refer to \citep{TroyLeeThesis} for an extensive discussion on the similarities and differences between Kolmogorov complexity and its time-bounded counterparts.

 Going beyond the technical difficulties of employing time-bounded Kolmogorov complexity, which some papers such as \citep{DBLP:conf/stoc/Hirahara21} managed to overcome with the right assumptions in place, there is perhaps a more relevant issue in the application of notions such as $\Kt$ and $\K^t$ to algorithms and complexity: these classical measures refer to \emph{deterministic} algorithms and programs. However, in many cases it is desirable or even necessary to consider \emph{randomised} algorithms. Since the random strings that are part of the input of a randomised algorithm have high complexity, the classical theory of time-bounded Kolmogorov complexity might be inappropriate or simply cannot be applied in such contexts.

To mitigate these issues and develop a more robust theory of time-bounded Kolmogorov complexity that can be deployed in the important setting of randomised computations, some recent papers \citep{DBLP:conf/icalp/Oliveira19, DBLP:conf/icalp/LuO21, LOS21, GKLO22, LOZ22} have explored \emph{probabilistic} notions of time-bounded Kolmogorov complexity. For this to make sense, we must conciliate the high complexity of a random string, which can be accessed by a randomised algorithm, with the goal of obtaining a succinct representation of $x \in \{0,1\}^*$. Note that simply storing a good choice of the random string $r$ for a small program $M$ that prints $x$ when given $r$ does not lead to a succinct representation of $x$.

The key concept employed in the aforementioned papers is that of a \emph{probabilistic representation} of the string $x$. In other words, this is the code of a randomised program $M$ such that, for most choices of its internal random string $r$, $M$ prints $x$ from $r$. Observe that the representation itself is a deterministic object: the code of $M$. However, to recover $x$ from $M$, we must run the randomised algorithm $M$, meaning that we obtain $x$ with high probability but there might be a small chance that $M$ outputs a different string.\footnote{This is similar to the notion of a pseudodeterministic algorithm from \citep{DBLP:journals/eccc/GatG11}.} If $|M|$ is small, we obtain a succinct probabilistic representation of $x$. It is possible to introduce different variants of probabilistic time-bounded Kolmogorov complexity, and we properly define them in \Cref{sec:prob_notions}.

	The investigation of probabilistic Kolmogorov complexity and of probabilistic representations is motivated from several angles:
	\begin{itemize}
	    \item [(\emph{i})] 	If we are  running a randomised algorithm over an input string $x$, then storing a probabilistic representation of $x$ instead of $x$ can be done without loss of generality. There is already a small probability that the randomised algorithm outputs an incorrect answer, so it makes sense to tolerate a small probability of computing over a wrong input as well (i.e., when $x$ is not correctly recovered from its probabilistic representation). 
	    \item [(\emph{ii})] 	We will see later in the survey that probabilistic Kolmogorov complexity allows us in some cases to obtain \emph{unconditional} versions of results that previously were only known to hold under strong complexity-theoretic assumptions. 
	    \item [(\emph{iii})] 	
	As alluded to above, there are situations where the deterministic time-bounded measures simply cannot be applied due to the presence of randomised computations involving random strings of high complexity.
	    \item [(\emph{iv})]  	
	Finally, advances in probabilistic Kolmogorov complexity can be translated into results and insights for the classical notions of $\Kt$ complexity and $\K^t$ complexity, under certain derandomisation hypotheses. 
	\end{itemize}

Before describing our results and explaining the points mentioned above in more detail, we present a list of five fundamental questions to guide our investigation and exposition of probabilistic Kolmogorov complexity.\\

	\noindent {\large \textbf{Q1.}} \textbf{Usefulness:} \emph{Are there shorter probabilistic representations for natural objects, such as prime numbers? Can such representations detect structure in data that is inaccessible for $\Kt$ and $\K^t$?}\\
	
	\vspace{-0.2cm}
	
	\noindent {\large \textbf{Q2.}} \textbf{Probabilistic Compression:} \emph{If succinct probabilistic representations exist, how can we efficiently compute one such representation? This is particularly relevant for data compression.}\\
	
		\vspace{-0.2cm}
	
			\noindent {\large \textbf{Q3.}} \textbf{Applications:} \emph{Are there interesting  applications of probabilistic time-bounded Kolmogorov complexity to algorithms and complexity theory?}\\
			
				\vspace{-0.2cm}
			
	\noindent {\large \textbf{Q4.}} \textbf{Computational Hardness:} \emph{If provably secure cryptography exists, it must be impossible to efficiently detect certain patterns in data. Is it computationally hard to decide if a string admits a succinct probabilistic  representation?}\\
	
		\vspace{-0.2cm}
	
		\noindent {\large \textbf{Q5.}} \textbf{Finding an Incompressible String:} \emph{Can we explicitly produce a string that does not admit a short probabilistic representation? What are such strings useful for?}\\
	
	In the remaining parts of this article, we explain the recent progress on Questions Q1-Q5 achieved by references \citep{DBLP:conf/icalp/Oliveira19, DBLP:conf/icalp/LuO21, LOS21, GKLO22, LOZ22}. Along the way, we highlight some concrete open problems and present directions for further research. Due to space constraints, we often provide only a sketch of the underlying arguments, referring to the original references for more details.\\

	\noindent \textbf{Organisation and Overview.} For convenience of the reader, we provide below a brief overview of each remaining section of this survey and how it relates to Questions Q1-Q5 described above.\\
	
	\noindent -- \Cref{sec:basic_definitions} fixes notation and formalises the deterministic time-bounded Kolmogorov complexity notions $\Kt$ and $\K^t$.\\
	
	\vspace{-0.2cm}
	
	\noindent -- \Cref{sec:prob_notions} formalises the intuitive concept of probabilistic representations discussed above. We introduce the probabilistic measures $\rKt$, $\rK^t$, and $\pK^t$ and describe some simple applications.\\ 
	
	\vspace{-0.2cm}
	
	\noindent -- \Cref{sec:primes_PRGs} addresses Question Q1 (Usefulness) and explains  a  result from \citep{LOS21} showing that infinitely many primes admit efficient probabilistic representations of sub-polynomial complexity. This is a significant improvement over the aforementioned $\approx n/2$ bound for $\Kt$ complexity.\\
	
	\vspace{-0.2cm}
	
	\noindent -- \Cref{sec:sampling_coding} covers the relation between sampling algorithms for a distribution over strings and the existence of probabilistic representations for individual strings \citep{DBLP:conf/icalp/LuO21, LOZ22}. Such results are called source coding theorems and have applications to Question Q2 (Probabilistic Compression).\\
	
		\vspace{-0.2cm}
	
	\noindent -- \Cref{sec:applications} approaches Question Q3 (Applications) and discusses applications of $\rKt$, $\rK^t$, and $\pK^t$ to average-case complexity and learning \citep{GKLO22, LOZ22}. We employ these notions to simplify previous proofs, obtain new results that crucially rely on probabilistic Kolmogorov complexity, and establish unconditional analogues of theorems that were only known under derandomisation hypotheses.\\
	
		\vspace{-0.2cm}
	
	\noindent -- \Cref{sec:prob_vs_det} is connected to Question Q1 (Usefulness) and focuses on the relation between time-bounded deterministic and probabilistic measures. We observe that these notions essentially coincide under strong enough derandomisation assumptions \citep{DBLP:conf/icalp/Oliveira19, GKLO22}. Assuming them, insights from probabilistic Kolmogorov complexity readily translate into information about $\Kt$ and $\K^t$.\\
	
		\vspace{-0.2cm}
	
	\noindent -- \Cref{sec:hardness} sheds light on Question Q4 (Computational Hardness) by unconditionally establishing that certain computational problems about estimating the probabilistic time-bounded Kolmogorov complexity of an input string cannot be solved in probabilistic polynomial time \citep{DBLP:conf/icalp/Oliveira19, LOS21}.\\
	
		\vspace{-0.2cm}
	
	\noindent -- \Cref{sec:explicit_construction} shows that Question Q5 (Finding an Incompressible String) is closely related to the existence of hierarchy theorems for probabilistic time \citep{DBLP:conf/icalp/LuO21, LOS21}, a fundamental question in computational complexity theory.\\
	
		\vspace{-0.2cm}
	
	\noindent -- \Cref{sec:final} provides some concluding remarks and prospects for the potential impact of (probabilistic) time-bounded Kolmogorov complexity in algorithms and complexity.\\

	\noindent \textbf{Acknowledgements.} We thank Michal Kouck\'{y} for the invitation to write this survey. We are grateful to Eric Allender, Bruno P.~Cavalar, Lijie Chen, Valentine Kabanets, Michal Kouck\'{y}, Ninad Rajgopal, and Marius Zimand for sharing comments and suggestions on a preliminary version of the text. This work received support from the Royal Society University Research Fellowship URF$\setminus$R1$\setminus$191059 and from the EPSRC New Horizons Grant EP/V048201/1.

	\section{Preliminaries}\label{sec:basic_definitions}

	For a positive integer $m$, we let $[m] \eqdef \{1, 2, \ldots, m\}$.  Given a non-negative real number $\alpha$, we let $\lceil \alpha \rceil \in \mathbb{N}$ denote the smallest integer $a$ such that $\alpha \leq a$. For a string $w \in \{0,1\}^*$, we use $|w| \in \mathbb{N}$ to denote its length. We let $\epsilon$ represent the empty string.

	Let $U$ be a Turing machine. For a function $t \colon \mathbb{N} \to \mathbb{N}$ and a string $x \in \{0,1\}^*$, we let 
	\[
    \K_U^t(x) \eqdef \min_{p \in \{0,1\}^*} \Big \{|p| \,\mid\, U(p, \epsilon)~\textnormal{outputs $x$ in at most $t(|x|)$ steps} \Big \}
  \]
  be the $t$-\emph{time-bounded Kolmogorov complexity} of $x$.
  The machine $U$ is said to be {\em time-optimal} if for every machine $M$ there exists a constant $c_M$ such that for all $x \in \{0,1\}^n$ and $t \colon \mathbb{N} \to \mathbb{N}$ satisfying $t(n) \geq n$,
  \[
    \K^{c_M \cdot t\log t}_U(x) \le \K^{t}_M(x) + c_M,
  \]
where for simplicity we write $t = t(n)$. It is well known that there exist time-optimal machines (see, e.g.,~\citep[Chapter 7]{li-vit:b:kolmbook-four}). We fix such a machine, and drop the index $U$ when referring to time-bounded Kolmogorov complexity measures.

Given strings $x, y \in \{0,1\}^*$, we can also consider the \emph{conditional} $t$-\emph{time-bounded Kolmogorov complexity of} $x$ \emph{given} $y$, defined as
\[
    \K^t(x \mid y) \eqdef \min_{p \in \{0,1\}^*} \Big \{|p| \,\mid\, U(p,y)~\textnormal{outputs $x$ in at most $t(|x|)$ steps} \Big \}.
  \]
  
  In the definitions above, the function $t \colon \mathbb{N} \to \mathbb{N}$ is fixed in advance. In many situations, it is also useful to consider a notion of time-bounded Kolmogorov complexity where the time bound of the machine is not fixed but instead affects the resulting complexity measure. One of the most prominent such measures is Levin's $\Kt$ complexity, defined as
  \[
  \Kt(x) \eqdef \min_{p \in \{0,1\}^*,\,t \in \mathbb{N}} \Big \{|p| + \lceil \log t \rceil \,\mid\, U(p,\epsilon)~\textnormal{outputs $x$ in at most $t$ steps} \Big \}.
  \]
  This definition can be extended to conditional $\Kt$ complexity $\Kt(x \mid y)$ in the natural way.

  From now on, we will not distinguish between a Turing machine $M$ and its encoding $p_M$ according to $U$. While the running time $t$ of $M$ on an input $y$ and the running time of the universal machine $U$ on $(p_M,y)$ might differ by a multiplicative factor of $O(\log t)$, this will be inessential in all results and applications discussed in this survey.\footnote{It is also possible to consider prefix-free notions of Kolmogorov complexity. Since our results hold up to additive $O(\log |x|)$ terms, we will not make an explicit distinction.} 
  
  We use $\K(x)$ to refer to the (time-unbounded) Kolmogorov complexity of the string $x$.

	\section{Probabilistic Notions of Kolmogorov Complexity: \texorpdfstring{$\rKt$}{rKt}, \texorpdfstring{$\rK^t$}{rK\^t}, and \texorpdfstring{$\pK^t$}{pK\^t}}\label{sec:prob_notions}

     In	\Cref{sec:basic_definitions}, we introduced two \emph{deterministic} notions of time-bounded Kolmogorov complexity: $\K^t$ and $\Kt$. In order to extend these definitions to the setting of \emph{randomised} computations, we consider an algorithm with a short description that outputs a fixed string $x \in \{0,1\}^n$ with high probability. Intuitively, the code of this algorithm serves as a \emph{probabilistic representation} of $x$.
	
	A bit more formally, we consider a randomised Turing machine (RTM) $M$ such that
	$$ \Pr_M[M(\epsilon)~\text{outputs}~x] \geq 2/3.$$ Since we are interested in time-bounded representations, in our definitions we must decide if we require (1) $M(\epsilon)$ to run in time $\leq t$ over all computation paths; or (2) with probability $\geq 2/3$, $M(\epsilon)$ runs in time $\leq t$ and outputs $x$. It turns out that this distinction is not really crucial for the results discussed in this survey, since they are robust to additive overheads of order $\log n$. In more detail, by specifying and storing a positive integer $i \in [n]$, which can be represented using just $\log n$ bits, we can always enforce the machine $M$ to stop in time $2^i$.
	
	\begin{remark}
	\emph{In the definitions presented below, we abuse notation and refer to a machine $M$ and its code. Formally, as in the definitions from the preceding section, $M$ should be an arbitrary string (and not be restricted to a string that is a well-formed description of a machine) that is provided as input to the machine $U$.\footnote{We assume that $U$ has access to a tape with random bits.} This is important to guarantee that the Kolmogorov complexity of an arbitrary string of length $n$ is at most $n + O(1)$. Defining Kolmogorov complexity and its time-bounded variants using the \emph{code} of a machine might only allow us to prove an upper bound of $O(n)$, which can create issues in some applications where a tight worst-case bound is needed. To simplify the presentation, we blur this distinction in the remaining parts of this survey.}
	\end{remark}

	\noindent \textbf{$\bm{\rK^t}$ Complexity} \cite{DBLP:journals/cc/BuhrmanLM05, LOS21}.\footnote{\cite{DBLP:journals/cc/BuhrmanLM05} refers to this notion as $\mathsf{CBP}^t$ complexity.} This is the randomised analogue of $\K^t$, where the time function $t \colon \mathbb{N} \to \mathbb{N}$ is fixed in advance. For a string $x \in \{0,1\}^*$, we let
    \[
	\rK^t(x) \eqdef \min_{\text{RTM}\,M} \left\{|M|  \,\mid\, M(\epsilon)~\text{outputs}~x~\text{in}~t(|x|)~\text{steps with probability} \geq 2/3\right\}
	\]
	denote its \emph{randomised} $t$\emph{-time-bounded Kolmogorov complexity}. As an example of the use of $\rK^t$, suppose a computationally unbounded party $A$ holds a string $x$, and that $A$ would like to communicate $x$ to a $t$-time-bounded party $B$ that has access to random bits. Then $A$ can send $k = \rK^t(x)$ bits to $B$ by communicating the description of a randomised Turing machine $M$ as above. $B$ is able to recover $x$ from $M$ with high probability simply by running $M(\epsilon)$.\\

\noindent \textbf{$\bm{\pK^t}$ Complexity} \cite{GKLO22}.	Fix a function $t \colon \mathbb{N} \to \mathbb{N}$, as before. For a string $x\in\bool^*$, the \emph{probabilistic} $t$\emph{-time-bounded Kolmogorov complexity} of $x$  is defined as
		\[
		\pK^{t}(x) \eqdef \min\left\{k \in \mathbb{N} \,\,\,\middle\vert\,\,\, \Prob_{w\sim\bool^{t(|x|)}}\left[\text{$\exists\,\text{TM}~M\in\bool^{k}$,  $M(w)$  outputs $x$ within $t(|x|)$ steps} \right] \geq \frac{2}{3}\right\}.
		\]
Note that $M$ is a deterministic machine in the above definition. In other words, if $k = \pK^{t}(x)$, then with probability  at least $2/3$ over the choice of the random string $w$, given $w$ the string $x$ admits a $t$-time-bounded encoding of length $k$, i.e., $\K^t(x \mid w) \leq k$. In particular, if two parties share a typical \emph{public} random string $w$, then $x$ can be transmitted with $k$ bits and decompressed in time $t = t(|x|)$. For a reader familiar with standard complexity classes, the condition $\K^t(x) \leq s$ is reminiscent of $\NP$, while $\rK^t(x) \leq s$ and $\pK^t(x) \leq s$ essentially correspond to $\MA$ and $\AM$, respectively.\\

The definition of $\pK^t$ complexity is more subtle than the definitions of $\K^t$ and $\rK^t$. In particular, small $\pK^t$ complexity provides a short efficient description only in the presence of a fixed, ``good'' random string. Interestingly, $\pK^t$ turns out to be surprisingly useful in applications of time-bounded Kolmogorov complexity, as discussed in Sections \ref{sec:sampling_coding} and \ref{sec:applications}.

The following inequalities immediately follow from these definitions.

\begin{fact}\label{fact:trivial_ineq}
For every string $x \in \{0,1\}^*$ and function $t \colon \mathbb{N} \to \mathbb{N}$, we have $\pK^t(x) \leq \rK^t(x) \leq \K^t(x)$.
\end{fact}
	
	\vspace{0.2cm}
	
\noindent \textbf{$\bm{\rKt}$ Complexity} \cite{DBLP:conf/icalp/Oliveira19}. We can also consider the \emph{randomised} $\Kt$ \emph{complexity} of a string $x \in \{0,1\}^*$, defined as
    \[
	\rKt(x) \eqdef  \min_{\text{RTM}\,M, \;t \in \mathbb{N}} \left\{|M| + \lceil \log t \rceil \mid M(\epsilon)~\text{outputs}~x~\text{in}~t~\text{steps with probability} \geq 2/3\right\}.
	\]
	
	\vspace{0.1cm}

All these probabilistic notions of time-bounded Kolmogorov complexity can be generalised to capture the conditional complexity of $x$ given $y$ in the natural way. As a concrete example, suppose a Boolean formula $F(x_1, \ldots,x_n)$ admits a satisfying assignment $\alpha \in \{0,1\}^n$ such that $\rKt(\alpha \mid F) \leq k$. Then we can find in time $\widetilde{O}(2^k \cdot |F|)$ and with probability $\geq 2/3$ a satisfying assignment of $F$ by performing the following randomised computation: for each $i \in [k]$, enumerate all RTM $M$ of description length $i$, run $M(F)$ for at most $2^{k - i}$ steps, and output the first string $\beta \in \{0,1\}^n$ generated in one of the simulations such that $F(\beta) = 1$.

An important property of Kolmogorov complexity is that, by a simple counting argument, most strings of length $n$ are \emph{incompressible}, i.e., they do not admit representations of length noticeably shorter than $n$. Similarly, most strings do not admit succinct probabilistic representations, even in the presence of a  fixed advice string $y$.

\begin{proposition}[Incompressibility]\label{l:incompressible}
Let $n \geq 1$ and consider an arbitrary time bound $t(n)$. For each string $y \in \{0,1\}^*$, measure $C \in \{\rK^t, \pK^t, \rKt \}$, and integer $k \geq 1$, the following holds.
$$
\Pr_{x \sim \{0,1\}^n}\big [C(x \mid y) < n-k \big] \;=\; O(2^{-k}).
$$
\end{proposition}

\begin{proof}[Proof Sketch]
For $C \in \{\rK^t, \rKt\}$, the result follows from a simple counting argument, using that a valid probabilistic representation represents a single string (i.e., the success probability of printing the string is $\geq 2/3$, so it is uniquely specified given the machine). 

On the other hand, when $C = \pK^t$, we argue as follows. If a large fraction of $n$-bit strings $x$ have bounded $\pK^t$ complexity, by an averaging argument, there is a fixed choice of the random string $w \in \{0,1\}^{t(n)}$ such that, given $w$, a large fraction of the $n$-bit strings admit bounded descriptions for this choice of $w$ as the random string. We can then use a similar counting argument to show that this is contradictory. See \citep{GKLO22} for the details.  
\end{proof}

It is also possible to define $\mathsf{pKt}$ complexity, in analogy with the aforementioned definitions. However, since we are not aware of an interesting application of $\mathsf{pKt}$, we will not discuss it here.

Other notions of time-bounded Kolmogorov complexity involving randomised computations have been considered in the literature. For instance, \citep{DBLP:journals/cc/BuhrmanLM05} considers $\mathsf{CAM}^t$, a variant that combines randomness and nondeterminism. Due to space constraints, this survey will only cover $\rKt$, $\rK^t$, $\pK^t$ and their recent applications.   
	
	\section{Prime Numbers with Short Descriptions and Pseudodeterministic PRGs} \label{sec:primes_PRGs}

	As briefly discussed in \Cref{sec:intro}, an important question about prime numbers is whether they admit succinct representations, which is tightly connected to the fundamental problem of generating large primes deterministically. While this remains a notoriously difficult question to answer, we can still ask whether prime numbers admit succinct \emph{probabilistic} representations. Results for this question were recently obtained in \cite{DBLP:conf/stoc/OliveiraS17, LOS21}, by considering different notions of (time-bounded) randomised Kolmogorov complexity. 
	
	Before describing these results, we first note that it is impossible to compress \emph{every} prime, given the Prime Number Theorem, which asserts that the number of primes whose values are less than or equal to $N$ is roughly $N/\log N$. In particular, by a simple counting argument, this means that we cannot compress every $n$-bit prime to $o(n)$ bits. Therefore, here we ask whether there is an infinite sequence $\{p_m\}_{m \in\mathbb{N}}$ of increasing primes $p_m$ that admit non-trivial probabilistic representations. The first non-trivial result of this form was established for $\rKt$ complexity.
	
	\begin{theorem}[$\rKt$ Upper Bounds for Primes~\cite{DBLP:conf/stoc/OliveiraS17}]\label{t:rkt-primes}
		For every $\varepsilon > 0$, there is an infinite sequence $\{p_m\}_{m \geq 1}$ of increasing primes $p_m$ such that $\mathsf{rKt}(p_m) \leq |p_m|^\varepsilon$, where $|p_m|$ denotes the bit-length of $p_m$.
	\end{theorem}
	
	\Cref{t:rkt-primes} was proved via the construction of a \emph{pseudodeterministic pseudorandom generator}. Informally, a pseudorandom generator (PRG) is an efficient procedure mapping a short string (called \emph{seed}) to a long string, with the property that its output ``looks random'' to algorithms with bounded running time.\footnote{Unconditionally constructing such PRGs is tightly connected to the derandomisation of probabilistic algorithms. While this remains a longstanding open problem, there has been progress in designing \emph{pseudodeterministic} PRGs.}   A PRG $G$ is called \emph{pseudodeterministic} if there is a probabilistic algorithm that, given a seed $z$, computes $G(z)$ with high probability. The following pseudodeterministic PRG was obtained in \cite{DBLP:conf/stoc/OliveiraS17}.
    \begin{theorem}[A Pseudodeterministic Sub-Exponential Time PRG~\cite{DBLP:conf/stoc/OliveiraS17}] \label{t:os-PRG}~\\
		For every $\varepsilon > 0$ and $c,d \geq 1$, there exists a generator $G = \{G_n\}_{n \geq 1}$ with $G_n \colon \{0,1\}^{n^\varepsilon} \to \{0,1\}^n$ for which the following holds:		\begin{itemize}[leftmargin=*]
		    \item[] \emph{Running Time:} There is a probabilistic algorithm that given $n$,  $x \in \{0,1\}^{n^\varepsilon}$, runs in time $O\!\left(2^{n^{\varepsilon}}\right)$ and outputs $G_n(x)$ with probability $\geq 2/3$.
		    \item[] \emph{Pseudorandomness:} For every algorithm $A$ that runs in time at most $n^c$, there exist infinitely many input lengths $n$ such that 
		    \[
		    \left| \Pr_{x \sim \bool^{n}}[A(x) = 1] - \Pr_{z \sim \bool^{n^\varepsilon}}[A(G_n(z))= 1]   \right| \leq  \frac{1}{n^d}.
		    \]
		\end{itemize}
	\end{theorem}
	Assuming \Cref{t:os-PRG}, we show how to obtain \Cref{t:rkt-primes}.
	\begin{proof}[Proof of \Cref{t:rkt-primes}]
		Let $A$ be a deterministic polynomial-time algorithm for primality testing (e.g., \cite{Agrawal02primesis}), which takes as input an $n$-bit integer $x$ and outputs $1$ if and only if $x$ is a prime. Suppose $A$ runs in time $n^c$ for some constant $c>0$. Note that by the Prime Number Theorem, a uniformly random $n$-bit integer is a prime number with probability at least $1/O(n)$.
		
		Let $\varepsilon > 0$ be any constant, and consider an infinitely often pseudodeterministic PRG $\{G_n\}_{n}$ from \Cref{t:os-PRG} with $G_n \colon \{0,1\}^{n^{\varepsilon/2}} \to \{0,1\}^{n}$ that is secure against $\left(n^c\right)$-time algorithms and has associated error parameter $\gamma = 1/n^2$. By the second item of \Cref{t:os-PRG}, for infinitely many values of $n$, we have
		\[
			\left| \Prob_{x\sim\bool^n}\left[A(x)=1\right] - 			\Prob_{z\sim\bool^{n^{\varepsilon/2}}}\left[A(G_n(z))=1\right] \right| \leq \frac{1}{n^2},
		\]
		which implies
		\[
			\Prob_{z\sim\bool^{n^{\varepsilon/2}}}\left[A(G_n(z))=1\right]\geq \frac{1}{O(n)}-\frac{1}{n^2}\geq \frac{1}{O(n)}.
		\]
		In particular, this means that there exists some $z\in\bool^{n^{\varepsilon/2}}$ such that $p\vcentcolon=G(z)$ is an $n$-bit prime. By hardcoding $n$ and this seed $z$, and using that $G(z)$ is a uniform procedure that can be computed probabilistically in time $t(n) = O\!\left(2^{n^{\varepsilon/2}}\right)$, we get that for infinitely many values of $n$, there is an $n$-bit prime $p$ such that
		\[
			\rKt(p) \leq \left(n^{\varepsilon/2} + O(\log n) + O(1)\right) + \log \!\left(O\!\left(2^{n^{\varepsilon/2}}\right)\right) \leq n^{\varepsilon},
		\]
		as desired.
			\end{proof}

	For those primes shown to have small $\rKt$ complexity, given the corresponding encoding, one can probabilistically recover the prime in sub-exponential time. We can then further ask whether we can obtain succinct representations that can be decoded more efficiently, say, in polynomial time. Note that this is precisely to show that there are infinitely many primes whose $\rK^{\poly}$ complexity is small. This question was answered in the affirmative by a subsequent work of Lu, Oliveira and Santhanam.
	
	\begin{theorem}[$\rK^{\poly}$ Upper Bounds for Primes~\cite{LOS21}]\label{t:rk-poly-primes}
		For every $\varepsilon > 0$, there is an infinite sequence $\{p_m\}_{m \geq 1}$ of increasing primes $p_m$ such that $\mathsf{rK}^t(p_m) \leq |p_m|^\varepsilon$, where $t(n) = n^k$ for some constant $k = k(\varepsilon) \geq 1$, and $|p_m|$ denotes the bit-length of $p_m$.
	\end{theorem}
	Similar to \Cref{t:rkt-primes}, \Cref{t:rk-poly-primes} was proved via the construction of a certain pseudodeterministic PRG. Note that the reason why we got sub-exponential decoding time in \Cref{t:rkt-primes} is due to the fact that the PRG from \Cref{t:os-PRG} requires sub-exponential time to compute. Then to obtain a polynomial decoding time as in \Cref{t:rk-poly-primes}, it suffices to construct a (pseudodeterministic) PRG that can be computed in polynomial time. Such a PRG was obtained in \cite{LOS21} using a more sophisticated approach that builds on \cite{DBLP:conf/stoc/OliveiraS17}.
	
	\begin{theorem}[A Pseudodeterministic Polynomial-Time PRG with $1$ Bit of Advice~\cite{LOS21}] \label{t:los-PRG}~\\
		For every $\varepsilon > 0$ and $c,d \geq 1$, there exists a generator $G = \{G_n\}_{n \geq 1}$ with $G_n \colon \{0,1\}^{n^\varepsilon} \to \{0,1\}^n$ for which the following holds:
		\begin{itemize}[leftmargin=*]
		    \item[] \emph{Running Time:} There is a probabilistic polynomial-time algorithm that given $n$,  $x \in \{0,1\}^{n^\varepsilon}$, and an advice bit $\alpha(n) \in \{0,1\}$ that is independent of $x$, outputs $G_n(x)$ with probability $\geq 2/3$.
		    \item[] \emph{Pseudorandomness:} For every algorithm $A$ that runs in time at most $n^c$, there exist infinitely many input lengths $n$ such that 
		    \[
				\left| \Pr_{x \sim \bool^{n}}[A(x) = 1] - \Pr_{z \sim \bool^{n^\varepsilon}}[A(G_n(z))= 1]   \right| \leq  \frac{1}{n^d}.
			\]
		\end{itemize}
	\end{theorem}
	Using \Cref{t:los-PRG}, it is easy to show \Cref{t:rk-poly-primes} by mimicking the above proof of \Cref{t:rkt-primes}, with one caveat that computing the PRG in \Cref{t:los-PRG} requires one bit of advice. However, this extra bit can be hardcoded into the encoding without affecting its length by much.

		We remark that the results presented above work in much more generality, and can be used to show that any dense language decidable in polynomial time admits infinitely many positive inputs of sub-polynomial $\mathsf{rKt}^{\mathsf{poly}}$ complexity. The set of primes is just one interesting example of such a language. We refer to \cite{DBLP:conf/stoc/OliveiraS17, LOS21} for additional applications of pseudodeterministic PRGs and for the proofs of Theorems \ref{t:os-PRG} and \ref{t:los-PRG}.
	
	We end this section with a couple of open problems. Note that both \Cref{t:rkt-primes} and \Cref{t:rk-poly-primes} show only that there are \emph{infinitely many} values of $n$ such that some $n$-bit prime has $\rKt$ or $\rK^{\poly}$ complexity at most $n^{\varepsilon}$. 
	
	\begin{problem}\label{p:prime_1}
	Show that for each $\varepsilon > 0$, there exists $n_0$ such that for every $n \geq n_0$, there is an $n$-bit prime $p_n$ such that $\rKt(p_n) \leq n^{\varepsilon}$.
	\end{problem}
Also, can we improve the sub-polynomial upper bounds to, say, poly-logarithmic? 

	\begin{problem}\label{p:prime_2}
	Prove that there is a constant $C \geq 1$ and an infinite sequence $\{p_m\}_{m \geq 1}$ of increasing primes $p_m$ such that $\rKt(p_m) = (\log |p_m|)^C$.
	\end{problem}

	\section{Sampling Algorithms, Coding Theorems, and Search-to-Decision Reductions}\label{sec:sampling_coding}

	The coding theorem for Kolmogorov complexity roughly states that if a string $x$ can be sampled with probability $\delta$ by some algorithm $A$, then its Kolmogorov complexity $\K(x)$ is at most $\log(1/\delta)+O_A(1)$. In particular, strings that can be generated with non-trivial probability by a program of small description length admit shorter representations. The coding theorem is a fundamental result in Kolmogorov complexity theory that has found many applications in theoretical computer science (see, e.g., \cite{DBLP:journals/ipl/LiV92, TroyLeeThesis, DBLP:journals/mst/Aaronson14, DBLP:journals/eccc/IlangoRS21}). In fact, \cite{TroyLeeThesis} regards the coding theorem as one of the four pillars of Kolmogorov complexity.\footnote{The other three are incompressibility, language compression, and symmetry of information.}
	
	The proof of the coding theorem crucially explores the time-unbounded feature of the Kolmogorov complexity measure, and it is unclear how it can be extended to the time-bounded setting. Ideally, we would like to show that if a string $x$ can be generated with probability $\delta$ by some \emph{efficiently samplable} distribution, then its time-bounded Kolmogorov complexity $\Kt(x)$ is about $\log(1/\delta)$. One reason why such a time-bounded coding theorem is hopeful is that it can be proven under certain strong derandomisation assumption \cite{DBLP:conf/coco/AntunesF09}.\footnote{The assumption in \cite{DBLP:conf/coco/AntunesF09} states that there is a language $L\in {\sf TIME}\!\left[2^{O(n)}\right]$ that requires Boolean circuits of size $2^{\Omega(n)}$ for all but finitely many $n$, even in the presence of oracle gates to a $\Sigma^p_2$-complete problem in the circuit.} In particular, under such an assumption, if a polynomial-time samplable distribution outputs a string $x$ with probability at least $\delta$, then $\Kt(x)\leq \log(1/\delta)+O(\log n)$. However, the latter result is only \emph{conditional}, in the sense that it relies on an unproven assumption that seems far beyond the reach of currently known techniques. Moreover, strong assumptions of this form could even be false. While it remains unclear whether we can obtain a coding theorem for $\Kt$, \cite{DBLP:conf/icalp/LuO21} considered the problem of establishing an \emph{unconditional} coding theorem in the \emph{randomised time-bounded setting}. Somewhat surprisingly, it can be shown \emph{unconditionally} that if a string $x$ can be sampled efficiently with probability $\delta$, then $\rKt(x)\leq O(\log 1/\delta)+O(\log n)$. In a subsequent work \cite{LOZ22}, this result is further improved to $\rKt(x)\leq (2+o(1))\cdot \log 1/\delta+O(\log n)$.
	
	\begin{theorem}[Coding Theorem for $\rKt$~\cite{LOZ22}]\label{t:rkt-coding}
		Suppose there is an efficient algorithm $A$ for sampling strings such that $A(1^n)$ outputs a string $x \in \{0,1\}^n$ with probability at least $\delta$. Then
		\[
		\rKt(x) \;\leq\; 2 \log(1/\delta) + O\!\left(\log n+ \log^2 \log(1/\delta)\right),
		\]		
		where the constant behind the $O(\cdot)$ depends on $A$ and is independent of the remaining parameters. Moreover, given $x$, the code of $A$, and $\delta$, it is possible to compute in time $\mathsf{poly}(n,|A|)$, with probability $\geq 0.99$, a probabilistic representation of $x$ that satisfies this $\rKt$-complexity bound. \emph{(}The running time of this algorithm does not depend on the time complexity of $A$.\emph{)}
	\end{theorem}
	Similar to the results in the previous section that are concerned with the compressibility of prime numbers, the results of \cite{DBLP:conf/icalp/LuO21, LOZ22} again show the power of utilizing randomness in Kolmogorov complexity, which enables us to establish results for time-bounded Kolmogorov complexity that seem very difficult to show in the deterministic setting. We refer to these papers for a discussion of the techniques employed to show an unconditional coding theorem for $\rKt$. 
	
	We note that (as in previous work of \cite{DBLP:conf/icalp/LuO21}) the coding theorem in \Cref{t:rkt-coding} has an unexpected \emph{constructive} feature: it gives a polynomial-time probabilistic algorithm that, when given $x$, the code of the sampler, and $\delta$, outputs a probabilistic representation of $x$ that certifies the claimed $\rKt$ complexity bound. (Additionally, the running time of this algorithm does not depend on the running time of the sampler.) Such an \emph{efficient} coding theorem has interesting implications for search-to-decision reductions for $\rKt$. Recall that a search-to-decision reduction is an efficient procedure that allows one to find solutions to a problem from the mere ability to decide when a solution exists. Using results from  \cite{DBLP:conf/icalp/LuO21, LOZ22}, one can show the following search-to-decision reduction for $\rKt$.
	
	\begin{theorem}[Instance-Wise Search-to-Decision Reduction for $\rKt$~\cite{DBLP:conf/icalp/LuO21}]\label{t:search-to-decision-rKt}
		Let $\mathcal{O}$ be a function that linearly approximates $\rKt$ complexity. That is, for every $x \in \{0,1\}^*$,
		\[
		\Omega(\rKt(x)) \;\leq\; \mathcal{O}(x) \;\leq\; O(\rKt(x)).
		\]
		Then there is a randomised polynomial-time algorithm with access to $\mathcal{O}$ that, when given an input string $x$, outputs with probability $\geq 0.99$ a valid $\rKt$ representation of $x$ of complexity $O(\rKt(x))$. Furthermore, this algorithm makes a single query $q$ to $\mathcal{O}$, where $q = x$. 
	\end{theorem}
	
	\begin{proof}[Proof Sketch.] We would like to invoke \Cref{t:rkt-coding} to efficiently compute an $\rKt$ representation of $x$, but the ``moreover'' part of this result requires the explicit code of a sampler.  The idea is to construct a ``universal'' sampler that outputs $x$ with the desired probability, then to hit this sampler with an appropriate  coding theorem for $\rKt$. For simplicity, suppose we knew the exact value $k = \rKt(x) \in \mathbb{N}$. Consider the following sampler $A$:\\
	
	\vspace{-0.2cm}
	
	\noindent $A(1^n)$: Randomly selects a randomised program $M$ of length $k$ among all strings in $\{0,1\}^k$. Run $M$ for at most $2^k$ steps, then output the $n$-bit string that $M$ outputs during this simulation (or the string $0^n$ if $M$ does not stop or its output is not an $n$-bit string).\\
	
	\vspace{-0.2cm}
	
	Note that $A$ runs in time $t = \mathsf{poly}(n,2^k) = \mathsf{poly}(2^k)$ (since $k \geq \log n$ for any $n$-bit string), and that it outputs $x$ with probability at least $\delta = 2^{-k} \cdot 2/3$, since by the definition of $k$ at least one such program prints $x$ with probability at least $2/3$. By the coding theorem for $\rKt$ from \cite{DBLP:conf/icalp/LuO21} (which is stated in a slightly more general form than \Cref{t:rkt-coding}), one obtains that $\rKt(x) = O(\log(1/\delta)) + O(\log t) = O(k)$. Since $A$ is an explicit algorithm, crucially, its ``moreover'' part implies that we can efficiently output an $\rKt$-representation of $x$ of complexity $O(k)$. This completes the sketch of the proof.
	
	We refer to \cite[Section 4]{DBLP:conf/icalp/LuO21} for the formal proof of \Cref{t:search-to-decision-rKt}, which is a simple adaptation of the idea described here.
	\end{proof}

	An interesting feature of the above search-to-decision reduction is that it is \emph{instance-wise} in the sense that to produce
	a near-optimal $\rKt$ representation of $x$, we only need to make a single query to a decision oracle for $\rKt$ on the same $x$.\footnote{This is also called a \emph{search-to-profile} reduction in some references in Kolmogorov complexity \citep{DBLP:journals/sigact/RomashchenkoSZ21}.} Note that there are known search-to-decision reductions in the context of time-bounded Kolmogorov complexity with respect to various notions of complexity (e.g., \cite{DBLP:conf/coco/CarmosinoIKK16, DBLP:conf/focs/Hirahara18,  IlangoCCC20, Ilango-Loff-Oliveira-CCC20, DBLP:conf/focs/LiuP20, DBLP:conf/focs/Ilango21}), but they require an oracle to the decision problem that is correct on all or at least on a large fraction of inputs. As a consequence of this feature, we can easily derive the following result.\footnote{Results of this form were previously known in  time-unbounded Kolmogorov complexity (see \citep{DBLP:journals/cc/BauwensMVZ18}).}
	
		\begin{corollary}[``Short Lists with Short Programs''  \cite{DBLP:conf/icalp/LuO21}]\label{c:short_list}
		Given a string $x$ of length $n$, it is possible to compute with probability $\geq 0.99$ and in polynomial time a collection of at most $\ell = \log (n)$ strings $M_1, \ldots, M_\ell$ such that at least one of these strings is a valid $\rKt$ representation of $x$ of complexity $O(\rKt(x))$.
	\end{corollary}
	
	\begin{proof}
	We run the instance-wise search-to-decision reduction on the input $x$. While it is not clear how to efficiently estimate $\rKt(x)$,  we can still  ``guess'' the $\rKt$ complexity of $x$ to be of order $2^i$, for each $i \in \{1,2, \ldots, \log n\}$. We run the procedure on each possible guess, obtaining a list of strings $M_1, \ldots, M_\ell$, where $\ell = \log n$. Since there is at least one value $i$ such that $2^i = \Theta(\rKt(x))$, we have the guarantee that in this case the reduction outputs with probability at least $0.99$ a valid $\rKt$ representation of $x$ of similar complexity. Therefore, the list contains with probability at least $0.99$ a representation of the desired form.
	\end{proof}

	While the above coding theorem for $\rKt$ is a novel development after a long gap in an area with only conditional results, it has an important drawback: the $\rKt$ upper bound is at least $2 \log(1/\delta))$ and hence is \emph{sub-optimal}. In contrast, the bounds in the time-unbounded setting and in the conditional result of \cite{DBLP:conf/coco/AntunesF09} mentioned above have the form $\log(1/\delta)$. A natural question then is whether we can show a coding theorem for $\rKt$ with an optimal dependence on the probability parameter $\delta$, which is crucial in many applications of the result. It turns out that under a certain hypothesis about the security of cryptographic pseudorandom generators\footnote{The hypothesis states that there is a pseudorandom generator $G \colon \{0,1\}^{\ell(n)} \to \{0,1\}^n$, where $(\log n)^{\omega(1)} \leq \ell(n) \leq n/2$, computable in time $\poly(n)$ that is secure against {\em uniform algorithms} running in time $2^{(1-\Omega(1))\cdot \ell(n)}$. Note that every candidate PRG of seed length $\ell(n)$ can be broken in time $2^{\ell(n)} \cdot \mathsf{poly}(n)$ by trying all possible seeds. This hypothesis can be viewed as a cryptographic analogue of the well-known strong exponential time hypothesis (SETH) about the complexity of $k$-CNF SAT \cite{DBLP:journals/jcss/ImpagliazzoP01}.}, the $\rKt$ bound in \Cref{t:rkt-coding} is essentially optimal if we consider only coding theorems that are \emph{efficient}, i.e., where an $\rKt$ representation can be constructed in polynomial time regardless of the running time of the sampler. In particular, \cite{LOZ22} showed that in this case, there is no efficient coding theorem that can achieve a bound of the form $\rKt(x) \leq (2-o(1)) \cdot \log(1/\delta) + \poly(\log n)$. On the other hand, the conditional coding theorem for $\Kt$ in \cite{DBLP:conf/coco/AntunesF09} is not efficient. This leads to the following open problem on (unconditionally) showing an \emph{existential} coding theorem for $\rKt$ with optimal parameters.
		\begin{problem}
		Show that if there is an efficient algorithm $A$ for sampling strings such that $A(1^n)$ outputs a string $x \in \{0,1\}^n$ with probability at least $\delta$, then $\rKt(x) \leq \log(1/\delta) + \poly(\log n)$.
	\end{problem}

	To this point, we have mentioned the existence of an optimal coding theorem for time-unbounded Kolmogorov complexity and an optimal conditional coding theorem for $\Kt$ (in fact, the conditional result holds even for $\K^{t}$ for $t = \mathsf{poly}(n)$). Also, an unconditional coding theorem can be obtained for $\rKt$ but its dependency on the probability parameter $\delta$ is not $\log (1/\delta)$ (\Cref{t:rkt-coding}). Note that $\rKt$ can be viewed as a ``relaxed'' notion of $\Kt$ and is intermediate between $\K$ and $\Kt$. If we consider some further relaxed notion of time-bounded Kolmogorov complexity, can we show a coding theorem that is both unconditional and optimal? 
	
	Note that the time-bounded measure $\pK$ can be viewed as an intermediate notion between time-unbounded Kolmogorov complexity and time-bounded $\rK$. It turns out that $\pK^t$ admits an optimal coding theorem. 
	
\begin{theorem}[Coding Theorem for $\pK^t$~\cite{LOZ22}]\label{t:pkt_coding}
	Suppose there is a randomised algorithm $A$ for sampling strings such that $A(1^n)$ runs in time $T(n) \geq n$ and outputs a string $x \in \{0,1\}^n$ with probability at least $\delta > 0$. Then
	\[
	\pK^t(x) \,=\,  \log(1/\delta) +  O\!\left(\log T(n)\right),
	\]
	where $t(n) =\poly\!\left(T(n)\right)$ and the constant behind the $O(\cdot)$ depends on $|A|$ and is independent of the remaining parameters.
\end{theorem}

	The proof of \Cref{t:pkt_coding} is similar in spirit to that of the conditional coding theorem for $\K^{\poly}$ in \cite{DBLP:conf/coco/AntunesF09}. As an application of the latter, \cite{DBLP:conf/coco/AntunesF09} showed a \emph{conditional} characterisation of the worst-case running times of languages that are in average polynomial time over all samplable distributions. Using \Cref{t:pkt_coding}, \cite{LOZ22} provided an \emph{unconditional} characterisation, and this will be discussed in \Cref{sec:applications}.
	
	Finally, we can connect the time-bounded coding theorems discussed in this section to the compressibility of prime numbers discussed in the previous section, via the following equivalence.
	
	\begin{theorem}[Equivalence Between Samplability and Compressibility \cite{DBLP:conf/icalp/LuO21}]\label{t:samp_comp}~\\
			Let $\delta\colon\mathbb{N}\to[0,1]$ be a time-constructible function. The following statements are equivalent.
		\begin{itemize}
			\item [\emph{(}i\emph{)}] \emph{\textbf{Samplability.}} There is a randomised algorithm $A$ for sampling strings such that, for infinitely many \emph{(}resp.~all but finitely many\emph{)} $n$, $A(1^n)$ runs in time $(1/\delta(n))^{O(1)}$ and outputs an $n$-bit prime $q_n$ with probability at least $\delta(n)^{O(1)}$.
			\item [\emph{(}ii\emph{)}] \emph{\textbf{Compressibility.}} For infinitely many \emph{(}resp.~all but finitely many\emph{)} $n$, there is an $n$-bit prime $p_n$ with $\rKt(p_n) =  O(\log (1/\delta(n)))$.
		\end{itemize}
	\end{theorem}
	
	\begin{proof}[Proof Sketch.]
	The implication from (\emph{i}) to (\emph{ii}) relies on the existing coding theorem for $\rKt$. The other direction employs a universal sampler in the spirit of the proof of \Cref{t:search-to-decision-rKt} sketched above. See \cite{DBLP:conf/icalp/LuO21} for the details.
	\end{proof}
	
	\Cref{t:samp_comp} can be seen as an analogue of the relation between deterministically constructing large primes and obtaining $\Kt$ upper bounds for primes, which was explained in \Cref{sec:intro}. Using this result, the problem of showing that prime numbers have smaller $\rKt$ complexity (Problem \ref{p:prime_2}) can be reduced to showing the existence of a faster sampling algorithm for primes. In particular, if we can sample an $n$-bit prime $p_n$ in time $2^{\poly(\log n)}$ with probability at least $2^{-\poly(\log n)}$, then $\rKt(p_n)\leq \poly(\log n)$.

	We remark that an even tighter equivalence between samplability and compressibility can be established using $\pK^t$ complexity, thanks to the optimality of \Cref{t:pkt_coding}. 
	
    \section{Applications to Average-Case Complexity and Learning Theory}\label{sec:applications}

    Understanding the relation between the average-case complexity of $\NP$ and its worst-case complexity is a central problem in complexity theory. More concretely, if every problem in $\mathsf{NP}$ is easy to solve on average, can we solve $\NP$ problems in polynomial time in the worst case? While addressing this question remains a longstanding open problem, significant results have been achieved in recent years using techniques from time-bounded Kolmogorov complexity \citep{DBLP:conf/focs/Hirahara20, DBLP:conf/stoc/Hirahara21, DBLP:conf/innovations/0001HV22} (see \cite{DBLP:journals/eatcs/Hir22} for an overview). Related techniques have also led to the design of faster learning algorithms under the assumption that $\NP$ is easy on average \citep{HN21}. Interestingly, the problems investigated in these references make no reference to Kolmogorov complexity. Still, the corresponding proofs rely on $\K^t$ complexity and its properties in important ways.
    
    	In this section, we describe recent applications of $\pK^t$ complexity to average-case complexity and learning theory \citep{GKLO22, LOZ22}. While the definition of $\pK^t$ is more subtle compared with $\K^t$ and $\rK^t$, its use comes with important benefits. As we explain later in this section, depending on the context, $\pK^t$ complexity allows us to extend previous results to the important setting of randomised computations, significantly simplify an existing proof, or obtain an unconditional result.\\
    	
    	\vspace{-0.2cm}
	
	\noindent \textbf{Average-Case Complexity.} We first review some standard definitions from average-case complexity theory (see \cite{DBLP:journals/fttcs/BogdanovT06} for a survey of this area). Recall that $D = \{D_n\}_{n \geq 1}$, where each $D_n$ is a distribution supported over $\{0,1\}^*$, is called an ensemble of distributions. We say that $D \in \mathsf{PSamp}$ (or $D$ is \emph{$\P$-samplable}) if there is a randomised polynomial-time algorithm $A$ such that, for every $n \geq 1$, $A(1^n)$ is distributed according to $D_n$.

				Let $D$ be an ensemble of distributions. We say that a language $L$ is solvable in \emph{polynomial time on average} with respect to $D$ if  there is a deterministic algorithm $A$ such that, for every $n$ and for every $x$ in the support of $D_n$, $A(x;n) = L(x)$, and there is a constant $\varepsilon > 0$ such that $\Exp_{x \sim D_n}[t_{A,n}(x)^\varepsilon/n] = O(1)$, where $t_{A,n}(x)$ denotes the running time of $A$ on input $(x;n)$. We remark that this is equivalent to the existence of a deterministic algorithm $B$ and of a polynomial $p$ such that the following conditions hold:
				\begin{itemize}
				    \item For every $n$, $\delta > 0$, and string $x$ in the support of $D_n$, $B(x; n, \delta)$ outputs either $L(x)$ or the failure symbol $\bot$;
				    \item For every $n$, $\delta > 0$, and every string $x$ in the support of $D_n$, $B(x; n, \delta)$ runs in time at most $p(n, 1/\delta)$;
\item  For every $n$ and every $\delta > 0$,
$$
\Pr_{x \sim D_n}[B(x; n,\delta) = \bot] \leq \delta.
$$
				\end{itemize}				
				We refer to \cite{DBLP:journals/fttcs/BogdanovT06} for more information about this definition and its motivation. 
				
		A pair $(L, D)$ is a \emph{distributional problem} if $L \subseteq \{0,1\}^*$ and $D$ is an ensemble of distributions. For a complexity class $\mathfrak{C}$ (e.g., $\mathfrak{C} = \NP$), we let $\Dist \mathfrak{C}$ denote the set of distributional problems $(L,D)$ with $L \in \mathfrak{C}$ and $D \in \mathsf{PSamp}$.
		We say that $(L,D) \in \Avg \P$ if $L$ is solvable in polynomial time on average with respect to $D$. 
		
		Note that in the equivalent definition of $\Avg \P$ the deterministic algorithm is never incorrect on an input $x$ in the support of the distribution. Similarly, it is possible to consider average-case complexity with respect to \emph{randomised errorless heuristic schemes}. Roughly speaking, such randomised algorithms are allowed to sometimes output the wrong answer, provided that on every input $x$ in the support of the distribution, the fraction of random strings for which the algorithm outputs the
wrong answer is small compared to the fraction of random strings for which it outputs either the
right answer or the fail symbol $\bot$. Analogously to the definition of $\Avg \P$, if a distributional problem $(L,D)$ admits a randomised errorless heuristic scheme, we say that   $(L,D) \in \Avg \BPP$. We refer again to \cite{DBLP:journals/fttcs/BogdanovT06} for the precise definition of this class and for an extensive discussion of this notion and its extensions.\footnote{It is also possible to consider randomised algorithms that can sometimes be incorrect on an input $x$ with high probability over their internal randomness. This leads to the class $\mathsf{HeurBPP}$ of distributional problems. Relaxing some assumptions in this section to the setting of $\mathsf{HeurBPP}$ is an interesting  research direction (see, e.g.,~\citep{DBLP:conf/innovations/HiraharaS22}).}
	
	\subsection{Worst-Case Time Bounds for Average-Case Easy Problems} \label{sec:app_characterizing}

	Suppose that a language $L$ is average-case easy. That is, $L$ is solvable in deterministic polynomial time on average with respect to \emph{all} $\P$-samplable distributions. What can we say about the time needed to solve $L$ in the \emph{worst case}? In a beautiful work, Antunes and Fortnow \cite{DBLP:conf/coco/AntunesF09} \emph{characterised} the worst-case running time of such a language using the notion of \emph{computational depth} \citep{DBLP:conf/coco/AntunesFM01}. Here the computational depth of a string $x$ for a time bound $t$ is defined as the difference $\K^{t}(x)-\K(x)$. It was shown in \cite{DBLP:conf/coco/AntunesF09}, under a strong derandomisation assumption, that a language $L$ is average-case easy \emph{if and only if} it can be solved in time $2^{O\left(\mathsf{\K}^{\poly}(x) - \mathsf{K}(x)+\log(|x|)\right)}$ for every input $x \in \{0,1\}^*$. 
	The proof of this result crucially relied on the use of an \emph{optimal} coding theorem for $\K^{t}$. Since such a coding theorem is only known under a strong derandomisation assumption (see \Cref{sec:sampling_coding}), the aforementioned characterisation  is subject to the same unproven assumption.
	
	As also mentioned in \Cref{sec:sampling_coding}, it was observed in \cite{LOZ22} that an optimal coding theorem can be unconditionally proved for $\pK^t$ (\Cref{t:pkt_coding}). It turns out that such a coding theorem enables us to show an \emph{unconditional} version of Antunes and Fortnow's characterisation, where the worst-case running times for languages that are average-case easy can be characterised using a notion of \emph{probabilistic computational depth}.
	
	A key idea in the proof of this result is a notion of \emph{universal} distribution via $\pK^t$. More specifically, for a computable time bound function $t$, we define $\mathsf{m}^t$ to be the (semi-)distribution whose probability density function is $\mathsf{m}^t(x) \eqdef 2^{-\pK^t(x)-b \log|x|}$, where $b > 0$ is a large enough constant (that depends only on $t$).\footnote{The reason why we define $\mathsf{m}^t(x)$ this way instead of using just $2^{-\pK^t(x)}$ is to make sure that it forms a (semi-)distribution, i.e., that the sum of the probabilities is at most $1$. More specifically, for every $t$, there is some constant $b>0$ such that $\K(x)\leq \pK^{t}(x)+b \log|x|$ for every $x$ (see \cite[Lemma 32]{LOZ22}), so $\sum_{x\in\bool^*} 2^{-\pK^t(x)-b \log|x|}\leq \sum_{x\in\bool^*} 2^{-\K(x)}\leq 1$, where the second inequality follows from Kraft's inequality. (Formally, to apply Kraft's inequality we need to consider prefix-free encodings. This is not an issue here, as a large enough constant $b$ makes this possible.)}
	\begin{theorem}[Unconditional ``Worst-Case Time Bounds for Average-Case Easy Problems''~\cite{LOZ22}]\label{t:AF-pKt}
		The following conditions are equivalent for any language $L \subseteq \{0,1\}^*$.\footnote{In this statement and in its proof, we do not make a distinction between distributions and semi-distributions. (In a semi-distribution, the sum of the probabilities might add up to less than $1$.)}
		\begin{enumerate}
			\item \label{i:AF-1} For every ${\sf P}$-samplable distribution  $D$, $L$ can be solved in polynomial time on average with respect to $D$.
			
			\item \label{i:AF-2} For every polynomial $p$, $L$ can be solved in polynomial time on average with respect to $\mathsf{m}^{p}$.
			
			\item \label{i:AF-3}
			For every polynomial $p$, there exists a constant $c>0$ such that the running time of some algorithm that computes $L$ is bounded by $2^{O\left(\mathsf{\pK}^{p}(x)-\mathsf{K}(x)+c\log(|x|)\right)}$ for every input $x \in \{0,1\}^*$.
		\end{enumerate}
	\end{theorem}
	\begin{proof}[Proof Sketch]
        For simplicity, to sketch the proof of this theorem we will also consider the notion of average-case easiness with respect to \emph{single distributions} instead of ensembles of distributions,\footnote{An algorithm $A$ runs in polynomial time on average with respect to a (semi-)distribution $D$ if there exists a constant $\varepsilon$ such that, $\sum_{x\in\bool^*} \frac{t_A(x)^{\varepsilon}}{|x|} D(x) \leq O(1)$, where $t_A(x)$ denotes the running time of $A$ on input $x$.} which does not incur a loss of generality (see \cite[Section 6]{DBLP:journals/fttcs/BogdanovT06}).
	
    	We first sketch the equivalence between Item 1 and Item 2. We need to show that the class of distributions $\mathsf{m}^{\poly}$ is ``universal'' for  the class of $\mathsf{P}$-samplable distributions, in the sense that a language $L$ is polynomial-time on average with respect to $\mathsf{m}^{\poly}$ if and only if the same holds with respect to all $\mathsf{P}$-samplable distributions. 
		Recall that if a distribution $D$ \emph{dominates} another distribution $D'$ (i.e., $D(x) \gtrsim D'(x)$ for all $x$) and $L$ is polynomial-time on average with respect to $D$, then the same holds with respect to $D'$. Therefore, to show the ``universality'' of $\mathsf{m}^{\poly}$, it suffices to establish the following claims.
		\begin{enumerate}
			\item Every $\mathsf{P}$-samplable distribution is dominated by $\mathsf{m}^{p}$, for some polynomial $p$.
			\item For every polynomial $p$, $\mathsf{m}^p$ is dominated by some $\mathsf{P}$-samplable distribution.
		\end{enumerate}
		
		The first item above says that for every $\P$-samplable $D$, $\mathsf{m}^{p}(x)\gtrsim D(x)$ for some polynomial $p$, which, by the definition of $\mathsf{m}^{p}$, means $\pK^p(x)\lesssim \log \!\left(1/D(x)\right)$. Note that this is essentially an optimal coding theorem for $\pK^{\poly}$ and hence follows from \Cref{t:pkt_coding}. To see the second item, consider any polynomial $p$. We define a $\P$-samplable distribution roughly as follows. We first pick $n$ with probability $\frac{1}{n\cdot (n+1)}$, and then randomly pick $k\in[2n]$, $w\in \bool^{p(n)}$, and a program $M\in\bool^{k}$. We then run $M(w)$ for at most $p(n)$ steps and output the string that $M$ outputs. It is easy to see that for every $x\in\bool^n$ , the above sampling process outputs $x$ with probability at least $2^{-\pK^p(x)}/n^{O(1)}$ and hence dominates $\mathsf{m}^p$.
		
		It remains to show the equivalence between Item 2 and Item 3. Here we describe the implication from Item 2 to Item 3, which highlights the use of a fundamental result in Kolmogorov complexity called \emph{Language Compression}. The other direction follows from a simple calculation (see \cite{LOZ22}).
		
		Consider the time bound $t$ described by an arbitrary polynomial $p$. Let $A$ be an algorithm that solves $L$ in polynomial time on average with respect to $\mathsf{m}^t$, and let $t_A(x)$ denote the running time of $A$ on input $x$.
		For $n,i,j\in\mathbb{N}$ with $i,j\leq n^2$, define
		\[
		S_{i,j,n} \eqdef \left\{x\in\bool^n \mid 2^i\leq t_A(x)\leq 2^{i+1} \text{ and } \pK^{t}(x)+b \log|x|=j \right\}.
		\]
		Consider a nonempty set $S_{i,j,n}$, and let $r \in \mathbb{N}$ be such that  $2^r\leq \left|S_{i,j,n}\right|< 2^{r+1}$. We claim that for every $x\in S_{i,j,n}$, its (time-unbounded) Kolmogorov complexity
		\begin{equation}\label{e:average_on_m_eq1_claim}
			\mathsf{K}(x)\leq r+O(\log n).
		\end{equation}
		To see this, note that given $i,j,n$, we can first enumerate all the elements in $S_{i,j,n}$, which can be done since $t$ is computable, and then using additional $r+1$ bits, we can specify $x$ in $S_{i,j,n}$. We remark that the core idea behind the above argument is the language compression theorem for (time-unbounded) Kolmogorov complexity, which states that for every (computable) language $L$, $\K(x)\leq \log \left|L\cap\bool^n\right| + O(\log n)$ for all $x\in L\cap\bool^n$.\footnote{In fact, it is possible to slightly modify the above argument, by  appropriately defining a language with slices in correspondence to the sets $S_{i,j,n}$, so that language compression can be applied directly.}
		
		Now fix any $n$ and $i,j\leq n^2$.  Let $r$ be such that  $2^r\leq \left|S_{i,j,n}\right| <  2^{r+1}$. Then by assumption and by the definition of $S_{i,j,n}$, we have for some constants $\varepsilon,d>0$,
		\[
		d \geq \sum_{x\in S_{i,j,n}} \frac{t_A(x)^{\varepsilon}}{|x|}\cdot \mathsf{m}^t(x) \geq 2^r \cdot \frac{2^{\varepsilon\cdot i}}{n}\cdot 2^{-j}= 2^{\varepsilon\cdot i +r-j-\log n},
		\]
		which yields $\varepsilon\cdot i +r-j-\log n \leq \log d$. 
		By \Cref{e:average_on_m_eq1_claim} and using $j = \pK^t(x) + b \log n$, this implies that for every $x\in S_{i,j,n} $,
		\[
		\varepsilon\cdot i \leq \pK^{t}(x) - \mathsf{K}(x) +  O(\log n).
		\]
		Therefore, we have that for every $x\in S_{i,j,n}$,
		\[
		t_A(x) \leq 2^{i+1}\leq 2^{\varepsilon^{-1}\cdot \left(\pK^{t}(x) - \mathsf{K}(x) + O(\log n)\right)} = 2^{O\left(\pK^{t}(x)-\mathsf{K}(x)+c \log(|x|)\right)},
		\]
		where $c>0$ is a large enough constant independent of $n = |x|$. Since it is not hard to see that every $x \in \{0,1\}^n$ is in some set $S_{i,j,n}$, the result follows.
	\end{proof}

	\subsection{Probabilistic Average-Case Easiness Implies Worst-Case Upper Bounds}\label{sec:app_ave_to_worst}
	
	The section covers recent developments from \citep{GKLO22}, which build on the breakthrough results of \citep{DBLP:conf/stoc/Hirahara21} and on the subsequent papers \citep{DBLP:conf/innovations/0001HV22, GK22, Hir22-symmetry}. In short, the results from \citep{DBLP:conf/stoc/Hirahara21} hold in the setting of \emph{deterministic} computations, while \citep{GKLO22} provides a framework that allows new relations between average-case complexity and worst-case complexity to be established in the more robust setting of \emph{randomised} computations. 
	
	Next, we provide a high-level exposition of some results from \citep{GKLO22} and their proofs. In particular, we explain the role of (conditional) versions of ``language compression'' and ``symmetry of information'' for $\pK^t$, and how $\pK^t$ turns out to be a complexity measure that is particularly well-suited for these applications (see \Cref{rem:whypKt} on ``Why $\pK^t$ complexity?'').
	
	Our goal is to show a \emph{worst-case} complexity upper bound for an arbitrary language $L \in \NP$ under an \emph{average-case} easiness assumption, such as $\Dist \NP \subseteq \Avg \P$ or the weaker $\Dist \NP \subseteq \Avg \BPP$. Note that \Cref{t:AF-pKt} naturally suggests an approach: if $L$ is easy on average (Item 1), then we can compute $L$ on every input $x \in \{0,1\}^*$ (Item 3) in time 
	$$2^{O\left(\mathsf{\pK}^{p}(x)-\mathsf{K}(x)\right)} \cdot \mathsf{poly}(|x|),$$
	where $p(\cdot)$ is a fixed but arbitrary polynomial. Therefore, if we could show that the quantity $\pK^p(x) - \K(x)$ is bounded for \emph{every} $x$, we would be done. (Note that this is indeed the case for a uniformly \emph{random} $x$, since $\pK^t(x)$ and $\K(x)$ are close to $n = |x|$ with high probability.)
	
	This is not possible, but we can still hope to adapt the proof of \Cref{t:AF-pKt} to obtain a more useful bound, under the assumption that $\Dist \NP \subseteq \Avg \BPP$. A closer inspection of the argument reveals that the value $\K(x)$ in the bound $\pK^p(x) - \K(x)$ comes from the use of \emph{language compression} for (time-unbounded) Kolmogorov complexity, which is applied to the sets $S_{i,j,n}$. If we had a language compression theorem for a time-bounded measure $\gamma$ (e.g.,~$\gamma = \K^t$), we would be able to derive a worst-case running time exponent of the form $\pK^p(x) - \gamma(x)$. This makes progress towards our goal, since $\gamma(x) \geq \K(x)$. This initial idea turns out to be feasible, for $\gamma = \pK^q$ (think of $q(\cdot)$ as a polynomial larger than $p(\cdot)$).

	\begin{theorem}[Language Compression for $\pK^t$ under {$\mathsf{DistNP} \subseteq \mathsf{AvgBPP}$}; Informal\footnote{For technical reasons, the actual formulation of this result considers an ensemble of promise problems with padded inputs of the form $(x,1^m)$, where $|x| = \ell(m)$. For simplicity, we omit this here. See \citep{GKLO22} for the precise statement.}] \label{t:language_compression}
	If $\mathsf{DistNP} \subseteq \mathsf{AvgBPP}$, then for every language $S\in\AM$, there is a polynomial $q$ such that for every $x\in S\cap \bool^n$,
	\[
	\pK^{q}(x) \leq \log \left|S \cap \bool^n\right| + \log q(n).
	\]
\end{theorem}

In order to implement the aforementioned plan, we need to make sure that the sets $S_{i,j,n}$ provide a language $S$ that is \emph{easy to compute}, since this is an assumption in \Cref{t:language_compression}. One can sidestep this issue by settling for a weaker result which assumes that the running time $t_A$ of the average-case algorithm on a given input can be efficiently estimated without running the algorithm. This notion leads to a class of distributional problems called $\Avg_{\mathsf{BPP}} \BPP$ in \citep{GKLO22}, and to the stronger initial assumption that $\Dist \NP \subseteq \Avg_{\mathsf{BPP}} \BPP$. Another crucial idea, which we will not cover in more detail here, is to prove that $\pK^t(y)$ can be efficiently estimated for every string $y$ under the assumption that $\NP$ is easy on average. We can then apply (an extension of) \Cref{t:language_compression} to appropriately modified sets $S'_{i,j,n}$, which yields a worst-case running time of the form
	$$2^{O\left(\mathsf{\pK}^{p}(x)-\mathsf{\pK}^q(x)\right)} \cdot \mathsf{poly}(|x|).$$

One could now hope for the quantity $\mathsf{\pK}^{p}(x)-\mathsf{\pK}^q(x)$, called the $(p,q)$-\emph{probabilistic computational depth} of $x$, to be bounded for every string $x$. While this is not clear for the polynomials $p(n)$ and $q(n)$, a simple but neat argument involving a telescoping sum \citep{DBLP:conf/stoc/Hirahara21, GKLO22} shows that, for any string $x$ of length $n$, for some time bound $t(n) \leq 2^{O(n/\log n)}$ we have $\pK^{t}(x) - \pK^{\poly(t)}(x) = O(n/\log n)$. Intuitively, if we could adapt the previous strategy so that it yields more general worst-case upper bounds involving $(t,\poly(t))$-\emph{probabilistic computational depth}, then a non-trivial exponent of $O(n/\log n)$ would be achieved by applying the argument to each choice of $t \leq 2^{O(n/\log n)}$.

A careful implementation of this plan leads to the following stronger consequence, where the worst-case upper bound holds for any language $L \in \AM$.

\begin{theorem}[\citep{GKLO22}]\label{t:avgBPPBPP}
If $\Dist \NP \subseteq \Avg_{\mathsf{BPP}} \BPP$, then $\AM \subseteq \mathsf{BPTIME}[2^{O(n/\log n)}]$.    
\end{theorem}

Can we obtain a similar worst-case upper bound under the weaker and more natural assumption that $\Dist \NP \subseteq \Avg \BPP$? (In other words, without assuming that the running time of the average-case algorithm can be efficiently estimated?) This is currently open. However, it is possible to prove the following implications, which can be seen as a strengthening of some results from \citep{DBLP:conf/stoc/Hirahara21} to the randomised setting. Recall that $\UP$ denotes the set of languages in $\NP$ whose positive instances admit unique witnesses.

    \begin{theorem}[Probabilistic Worst-Case to Average-Case Reductions~\cite{GKLO22}]\label{t:prob_reductions} The following results hold.
    	\begin{enumerate}
    		\item  If $\Dist\NP\subseteq \Avg\BPP$, then $\UP\subseteq\RTIME\!\left[2^{O(n/\log n)}\right]$. 
    		\item  If $\Dist\Sigma^{\P}_{2} \subseteq \Avg\BPP$, then $\AM \subseteq \BPTIME\!\left[2^{O(n/\log n)}\right]$. 
    		\item If $\Dist\PH\subseteq\Avg\BPP$, then $\PH\subseteq\BPTIME\!\left[2^{O(n/\log n)}\right]$.  
    	\end{enumerate}
    \end{theorem}

    The proof of \Cref{t:prob_reductions} relies on \emph{Symmetry of Information}, another pillar of Kolmogorov complexity (see \citep{TroyLeeThesis}). To describe a pair $(x,y)$ of strings, one can combine the most succinct representation of $x$ with the most succinct representation of $y$ when $x$ is given as advice. In Kolmogorov complexity, this is captured by the inequality $\K(x,y) \leq \K(x) + \K(y \mid x) + O(\log(|x| + |y|))$. The symmetry of information principle is a theorem in Kolmogorov complexity stating that this is essentially the most economical way of describing the pair $(x,y)$. In other words: $\K(x,y) \geq\K(x) + \K(y \mid x) - O(\log(|x| + |y|))$. One can then easily derive  that  $\K(x) - \K(x | y) = \K(y) - \K(y \mid x)$, up to a term of order $O(\log(|x| + |y|))$. Roughly speaking, the information that $x$ contains about $y$ is about the same the information that $y$ contains about $x$. 
    
  The proof of symmetry of information for $\K$ requires an exhaustive search, which is not available in the time-bounded setting. Nevertheless, different forms of the principle can still be established in this more delicate setting under average-case easiness assumptions \citep{GK22, Hir22-symmetry, GKLO22}.
    
     \begin{theorem}[Symmetry of Information for $\pK^t$ under {$\mathsf{DistNP} \subseteq \mathsf{AvgBPP}$}~\citep{GKLO22}\footnote{A more general version of this result is used by \citep{GKLO22} to establish \Cref{t:prob_reductions} and its extensions.}] \label{t:SoI}
     If $\Dist\NP \subseteq \Avg\BPP$, then there exist polynomials $p$ and $p_0$ such that for all sufficiently large $x, y \in \{0,1\}^*$ and every $t \geq p_0(|x|, |y|)$, 
	    	\[\pK^t(x,y) \geq \pK^{p(t)}(x) + \pK^{p(t)}(y \mid x) - \log p(t).\] 
    \end{theorem}
    
    Assuming \Cref{t:SoI}, we provide a high-level exposition of the proof of a variant of Item 2 from \Cref{t:prob_reductions}: If $\Dist \Sigma^{\mathsf{P}}_2 \subseteq \Avg \BPP$ then $\NP \subseteq \mathsf{RTIME}[2^{O(n/\log n)}]$. (A detailed informal presentation of Item 2 of \Cref{t:prob_reductions} can be found in \citep[Section 1.3]{GKLO22}.) Assume that $\Dist \Sigma^{\mathsf{P}}_2 \subseteq \Avg \BPP$, and let $L \in \NP$. Fix some $\NP$-verifier $V$ for this language. For a string $x \in L$ of length $n$, let $y_x$ be the lexicographic first string such that $V(x,y_x) = 1$.\\
    
    \vspace{-0.1cm}
    
    \noindent \textbf{1.} On the one hand, it follows from \Cref{t:SoI} that there is a universal constant $a \geq 1$ such that, for every large enough $t$, $\pK^{t^a}(y_x \mid x) \leq \pK^t(x,y_x)
     - \pK^{t^a}(x) + O(\log t)$.\\
    
        \vspace{-0.2cm}
    
\noindent \textbf{2.} On the other hand, under the assumption that $\Dist \Sigma^{\mathsf{P}}_2 \subseteq \Avg \BPP$, it is possible to prove that, for some universal constant $\varepsilon > 0$ and for every large enough $t$, $\pK^{t}(x,y_x) \leq \pK^{t^\varepsilon}(x) + O(\log t)$. This is non-trivial: while it is possible to recover $y_x$ from $x$ with a powerful enough oracle, we must obtain a description of the pair $(x,y_x)$ from (a fixed but arbitrary) $x$ without the aid of such an oracle, using only an average-case easiness assumption.\\

    \vspace{-0.2cm}
    
    \noindent \textbf{3.} Putting together the previous inequalities from Steps 1 and 2, we get that for every large enough $t$, $\pK^{t^a}(y_x | x) \leq \pK^{t^\varepsilon}(x) - \pK^{t^a}(x) + O(\log t)$. Consequently, we can upper bound $\pK^{t^a}(y_x | x)$ by the $(t^{\varepsilon}, t^a)$-probabilistic computational depth of $x$ plus $O(\log t)$, for any $t \geq \poly(n)$, where $n = |x|$.\\
    
    \vspace{-0.2cm}
    
    \noindent \textbf{4.} As in the proof sketch of \Cref{t:avgBPPBPP}, one can show that for every $x$ there is some $t(n) = 2^{O(n/\log n)}$ such that $\pK^{t^\varepsilon}(x) - \pK^{t^a}(x) = O(n/\log n)$. Consequently, using that $\pK^{t_1}(\cdot) \leq \pK^{t_2}(\cdot)$ if $t_1 \geq t_2$, there is a constant $C \geq 1$ such that, for every string $x$ of length $n$, we have $\pK^\gamma(y_x \mid x) \leq C \cdot n/\log n$, where $\gamma(n) = 2^{C \cdot n/\log n}$.\\
    
        \vspace{-0.2cm}
        
    \noindent \textbf{5.} Finally, given a positive instance $x$ of $L$ and the upper bound on  $\pK^\gamma(y_x \mid x)$ from Step 4, we can recover $y_x$ with probability $\geq 2/3$ in time $2^{O(n/\log n)}$. Indeed, this follows from the definition of conditional $\pK^t$ complexity: by sampling a random string $w$ of length $2^{C \cdot  n/\log n}$ and simulating all machines $M$ of length $\leq C \cdot n/\log n$ on input $(x,w)$ for at most $2^{C \cdot n/\log n}$ steps, we generate $y_x$ with probability at least $2/3$ over the choice of $w$. Since we can test each string produced in this way using the polynomial-time verifier $V(x, \cdot)$, it follows that $L \in \mathsf{RTIME}[2^{O(n/\log n)}]$.\\ 
   
    \vspace{-0.3cm}
    
    \begin{remark}[\textbf{Why $\pK^t$ complexity?}] \label{rem:whypKt} 
    \emph{Both language compression (\Cref{t:language_compression}) and symmetry of information (\Cref{t:SoI}) are established using techniques from  computational pseudorandomness related to the design and analysis of pseudorandom generators (PRGs). This approach has  proven extremely useful in time-bounded Kolmogorov complexity (see, e.g.,~\citep{DBLP:journals/siamcomp/AllenderBKMR06}). In a bit more detail, in the proof of both results we are interested in establishing bounds on the Kolmogorov complexity of a string $x$. A way of doing this is by considering the string $x$ as a source of ``hardness'' (e.g., view $x$ as a hard truth-table) in the construction of a generator $\mathsf{G}^x$. The typical analysis of a PRG provides a reconstruction routine, i.e., an algorithm implementing the proof that if we can break $\mathsf{G}^x$ using a distinguisher $D$, then $x$ cannot be hard. In other words, we obtain bounds on the conditional time-bounded Kolmogorov complexity of $x$ given $D$. Crucially, under assumptions such as $\Dist \NP \subseteq \Avg \BPP$, it is often possible to break the corresponding PRG $\mathsf{G}^x$. This provides a powerful way of analysing the time-bounded Kolmogorov complexity of strings in the context of Theorems \ref{t:avgBPPBPP} and \ref{t:prob_reductions}. More recently, the papers \citep{DBLP:conf/stoc/Hirahara20, DBLP:conf/stoc/Hirahara21} have highlighted the importance of a particular ``direct product'' generator $\mathsf{G}^x = \mathsf{DP}^x$, which has near-optimal ``advice'' complexity in its reconstruction procedure and provides tighter bounds on the complexity of $x$. In the \emph{randomised} reconstruction procedure of $\mathsf{DP}^x$, the advice depends on the particular choice of the random string employed by the procedure, which shows that for a noticeable fraction of random strings $w$, $x$ has a small description if we are given the random string $w$. Now observe that this corresponds precisely to $\pK^t$ complexity! In previous work \citep{DBLP:conf/stoc/Hirahara21}, this issue is not present because the stronger assumption that $\Dist \NP \subseteq \Avg \mathsf{P}$ provides near-optimal derandomisation \citep{DBLP:journals/mst/BuhrmanFP05} that allows one to directly get $\K^t$ bounds. However, the same PRG is not known to be available under the weaker assumption that $\Dist \NP \subseteq \Avg \BPP$.}
    \end{remark}
    
    As explained in \citep{GKLO22}, while previous works have employed  various techniques to \emph{remove} randomness from their arguments in order to analyze $\K^t$ complexity, the idea of \emph{incorporating} randomness in the framework (via $\pK^t$) comes with other benefits beyond the extension of results to the setting of randomised computations. For instance, \citep{DBLP:conf/innovations/0001HV22} established \emph{fine-grained} connections between worst-case and average-case complexity. Among other results, they showed that if $\mathsf{NTIME}[n]$ can be deterministically solved in quasi-linear time on average, then $\mathsf{UP} \subseteq \mathsf{DTIME}[2^{O(\sqrt{n\log n})}]$. While the argument from  \citep{DBLP:conf/innovations/0001HV22} requires the construction of an extremely fast PRG via a delicate analysis, the same result can be proved using $\pK^t$ complexity with a simpler proof \citep{GKLO22}.

As a potentially accessible direction, we pose the following problem related to \Cref{t:avgBPPBPP} and Item 1 of \Cref{t:prob_reductions}.

\begin{problem}
Show that if $\Dist \NP \subseteq \Avg \BPP$ then $\NP \subseteq \mathsf{BPTIME}[2^{O(n/\log n)}]$.
\end{problem}

	\subsection{Learning Algorithms from Probabilistic Average-Case Easiness}\label{sec:app_learning}

This section describes an application of probabilistic Kolmogorov complexity to computational learning theory. More precisely, we show that if $\Dist \NP \subseteq \Avg \BPP$, then polynomial-size Boolean circuits can be (agnostically) PAC learned under any samplable distribution in polynomial time. While it is not hard to learn general Boolean circuits under a \emph{worst-case} easiness assumption (e.g, $\NP \subseteq \BPP$) using Occam's razor (see, e.g.,~\citep{DBLP:books/daglib/0041035}), here we obtain an interesting consequence for learning under a \emph{weaker} average-case easiness assumption.

The proof adapts a similar learning result from \citep{HN21}, established under the assumption that $\Dist \NP \subseteq \Avg \P$ (i.e., average-case easiness for \emph{deterministic} algorithms). This exhibits a natural example of a result that can be lifted to the randomised setting with little effort via $\pK^t$ complexity.  
	
	 Let $\mathcal{C}$ be a class of Boolean functions. In the PAC learning model, a learner has access to examples $(x,f(x))$ labelled according to an unknown function $f \in \mathcal{C}$. The examples $x$ are drawn according to an unknown probability distribution $D_n$ supported over $\{0,1\}^n$. The goal of the learning algorithm is to produce, with high probability over its internal randomness and draw of labelled examples, a hypothesis $h$ such that $\Pr_{x \sim D_n}[h(x) \neq f(x)] \leq \varepsilon$.
	
	We say that the distribution $D_n \in \mathsf{Samp}[T(n)]/a(n)$ if it can be sampled by an algorithm that runs in time $T(n)$ and has advice complexity $a(n)$. (A sampler described by a uniform machine of code length $a$ counts as advice of length $a$.) We consider the learnability of the class $\mathcal{C} = \mathsf{SIZE}[s]$ of Boolean circuits of size at most $s(n)$, with respect to an unknown distribution $D_n$ from $\mathsf{Samp}[T(n)]/a(n)$. 
	
	As in \citep{HN21}, the result described below also holds in the more challenging setting of \emph{agnostic learning}, where the function $f$ only needs to be \emph{close} to some function in $\mathcal{C}$. (See \citep{GKLO22} for a concise presentation of this learning model.)

\begin{theorem}[Agnostic Learning from Probabilistic Average-Case Easiness of $\NP$~\cite{GKLO22}]\label{t:learning_from_easiness}~\\If $\mathsf{DistNP}\subseteq\mathsf{AvgBPP}$, then for any time constructible functions $s,T,a\colon\mathbb{N}\to \mathbb{N}$, and $\varepsilon\in[0,1]$, $\mathsf{SIZE}[s(n)]$ is agnostic learnable on $\mathsf{Samp}[T(n)]/a(n)$ in time
	\(
	\poly\!\left(n,\varepsilon^{-1},s(n),T(n), a(n)\right)
	\).
\end{theorem}

	For the proof of this result, the main idea is to design a         \emph{random-right-hand-side-refuter} (RRHS-refuter; see \citep{DBLP:conf/colt/Vadhan17, DBLP:conf/innovations/KothariL18}). In short, this is an algorithm that distinguishes the distribution $\left(x^{(1)}, \dots ,x^{(m)},f(x^{(1)}),\dots f(x^{(m)})\right)$, where each $x^{(i)}$ is picked from a fixed but unknown distribution $D_n$ and $f \in \caC$ is a fixed but unknown function, from the distribution $\left(x^{(1)}, \dots ,x^{(m)}, b^{(1)},\dots b^{(m)}\right)$, where in this case each $b^{(i)}$ is a uniformly random bit. It is known that such an algorithm can be converted into an agnostic learner for $\caC$ under the distribution $D_n$. 
	
	In \cite{HN21} an efficient RRHS-refuter is constructed using an algorithm that estimates the $\K^t$ complexity of a given string, which can be shown to exist under the assumption that $\Dist \NP\subseteq \Avg\P$ \cite{DBLP:conf/stoc/Hirahara21}. In more detail, \cite{HN21} proved that if a string is sampled from the first distribution, where $D_n$ is efficiently samplable and $f$ is computable by a polynomial size circuit, then it is likely to have bounded $\K^t$ complexity (for carefully chosen parameters $m$ and $t$). On the other hand, using symmetry of information and optimal coding for $\K^t$, which hold under an average-case easiness assumption \cite{DBLP:conf/stoc/Hirahara21}, it can be shown that a random string from the second distribution is likely to have large $\K^t$ complexity. 
	
	In contrast, under the weaker assumption that $\Dist \NP \subseteq \Avg \BPP$, we design an efficient algorithm that estimates the $\pK^t$ complexity of a given string, which is a more delicate measure than $\K^t$. Combining this algorithm with the symmetry of information for $\pK^t$ (\Cref{t:SoI}), which holds under the same probabilistic average-case easiness assumption, and the optimal coding result for $\pK^t$ (\Cref{t:pkt_coding}), we are able to construct in a similar way an efficient randomised RRHS-refuter. As before, this is sufficient to obtain the desired learning conclusion.
	
	It would be interesting to understand if under the same  average-case easiness assumption one can non-trivially learn general Boolean circuits with respect to an arbitrary distribution, i.e., in the standard sense of the PAC learning model.
	
	\begin{problem}
	Suppose that $\Dist \NP \subseteq \Avg \BPP$. Is it possible to PAC learn Boolean circuits of size $O(n)$ \emph{(}say, with error $\varepsilon = 1/10$\emph{)} in time $2^n/n^{\omega(1)}$?
	\end{problem}

We note that this would be possible (via Occam's Razor) if the same average-case easiness assumption led to stronger worst-case upper bounds for languages in $\NP$, such as the conclusion that $\NP \subseteq \BPTIME[2^{n^{0.499}}]$.

    	\section{Probabilistic Versus Deterministic Time-Bounded Kolmogorov Complexity}\label{sec:prob_vs_det}

    	We have seen that some questions that remain open for classical notions of time-bounded Kolmogorov complexity (such as $\Kt$) can be unconditionally answered in the case of $\rKt$, $\rK^t$, and $\pK^t$. For instance, we presented better bounds for primes with respect to $\rKt$ (\Cref{t:rkt-primes}) and $\rK^t$ (\Cref{t:rk-poly-primes}) in \Cref{sec:primes_PRGs}, and stated an optimal coding theorem for $\pK^t$ (\Cref{t:pkt_coding}) in \Cref{sec:sampling_coding}. Moreover, we exhibit several applications of probabilistic time-bounded Kolmogorov complexity to algorithms and complexity in \Cref{sec:applications}. It is perhaps a good point to discuss in more detail the relation between deterministic and probabilistic notions of Kolmogorov complexity.
    	
    	It turns out that, \emph{under strong enough derandomisation hypotheses}, for every string $x$, its deterministic and probabilistic time-bounded Kolmogorov complexities essentially coincide. For instance, for $\Kt$ and $\rKt$ we have the following relation.\footnote{Recall that $\mathsf{E} = \mathsf{DTIME}[2^{O(n)}]$ refers to the set of languages that can be decided in deterministic time $2^{O(n)}$, while $\mathsf{BPE} = \mathsf{BTIME}[2^{O(n)}]$ is the set of languages that can be decided in probabilistic time $2^{O(n)}$. The promise version of $\mathsf{E}$ is defined in the natural way. Recall that for $\promise\BPE$ we do not enforce the acceptance probability of the randomised machine to be bounded away from $1/2$ on inputs that do not satisfy the promise.}
    	
    	 \begin{theorem}[\cite{DBLP:conf/icalp/Oliveira19}]\label{t:Kt-and-rKt}
    		The following results hold.
    		\begin{itemize}
    			\item If $\promise\BPE\subseteq \promise\mathsf{E}$, then $\Kt(x)\leq O(\rKt(x))$ for every string $x$.
    			\item If $\Kt(x)\leq O(\rKt(x))$ for every string $x$, then $\BPE\subseteq \mathsf{E}/O(n)$.
    		\end{itemize}
    		In particular, $\rKt$ and $\Kt$ are linearly related measures if $\mathsf{E} \nsubseteq \mathsf{i.o.SIZE}\!\left[2^{\Omega(n)}\right]$.
    	\end{theorem}
    	
    	Note that the connection between derandomisation of probabilistic complexity classes and time-bounded Kolmogorov complexity holds in both directions: in a sense, collapsing $\Kt$ and $\rKt$ for every string $x$ (up to a constant multiplicative factor) is essentially equivalent to the derandomisation of (promise) $\BPE$, as stated in \Cref{t:Kt-and-rKt}.
    	
    	Similarly, under strong enough assumptions, we can show that $\pK^{\poly}(x)$, $\rK^{\poly}(x)$, and $\K^{\poly}(x)$ coincide up to an additive term of order $O(\log n)$.
    	
    		\begin{theorem}[\cite{GKLO22}]
    		The following results hold.
    		\begin{itemize}
    			\item If $\mathsf{E}\not\subseteq \mathsf{i.o.SIZE}\!\left[2^{\Omega(n)}\right]$, then there is a polynomial $p$ such that $\K^{p(t)}(x)\leq \rK^t(x)+ \log p(t)$, for every $n$-bit string $x$ and time bound $t(n) \geq n$.
    			\item If $\mathsf{E}\not\subseteq \mathsf{i.o. NSIZE}\!\left[2^{\Omega(n)}\right]$, then there is a polynomial $p$ such that $\K^{p(t)}(x)\leq \pK^t(x)+ \log p(t)$, for every $n$-bit string $x$ and time bound $t(n) \geq n$.
    			\item If $\BPE\not\subseteq \mathsf{i.o. NSIZE}\!\left[2^{\Omega(n)}\right]$, then there is a polynomial $p$ such that $\rK^{p(t)}(x)\leq \pK^t(x)+ \log p(t)$, for every $n$-bit string $x$ and time bound $t(n) \geq n$.
    		\end{itemize}
    	\end{theorem}
    	
    	\begin{proof} We describe the proof of the first item. The other two relations can be established by an appropriate modification of the argument, and we refer to \cite{GKLO22} for the details.
    	
    		Let $x \in \{0,1\}^n$, and let $t(n) \geq n$. First, the assumption $\mathsf{E}\not\subseteq \mathsf{i.o.SIZE}\!\left[2^{\Omega(n)}\right]$ implies that there is a PRG $G\colon\bool^{O(\log s)}\to\{0,1\}^s$ that ($1/s$)-fools size-$s$ Boolean circuits and has running time $\poly(s)$ \cite{IW97}. Suppose $\rK^t(x)\leq k$. Let $M\in\bool^{k}$ be a probabilistic machine  of running time at most $t$ that outputs $x$ with probability at least $2/3$.\footnote{If $M$ runs for more than $t$ steps on some computation path, we simply truncate its computation.} Consider the following function $C$ on inputs of length $t$:
		\[
			C(w) = 1 \iff M(w) = x.
		\]
		Clearly, $C$ can be implemented as a $\poly(t)$-size circuit. By definition, the acceptance probability of $C$ is at least $2/3$. Consequently, there is a seed $z\in\bool^{O(\log t)}$ such that $C(G(z))=1$, which in turn implies that $M(G(z))=x$. This means that, given the description of $M$ and $z$, we can \emph{deterministically} compute $x$ in time $\poly(t)$. In particular, $\K^{p(t)}\leq k + \log p(t)$, for some large enough polynomial $p(\cdot)$. This polynomial is selected as a function of the overhead in running time and description length caused by the PRG. For this reason, it does not depend on $x$ and $t$. This completes the proof.
    	\end{proof}
    
    As a consequence of these (conditional) equivalences, new insights about probabilistic time-bounded Kolmogorov complexity can also shed light on the classical deterministic notions.\footnote{As a concrete example, after proving \Cref{t:search-to-decision-rKt} in \citep{DBLP:conf/icalp/LuO21}, we noticed that a similar result also holds for $\Kt$, \emph{unconditionally}. See \citep{DBLP:conf/icalp/LuO21} for more information on this.} In particular, if one believes in the corresponding derandomisation assumptions, establishing certain results for $\rKt$, $\rK^t$, and $\pK^t$ can be seen as a necessary step before we are able to obtain similar statements for $\Kt$ and $\K^t$. One such example is the task of showing better upper bounds on the time-bounded Kolmogorov complexity of prime numbers (\Cref{sec:primes_PRGs}). 
    
    Of course, one of the main advantages of probabilistic time-bounded Kolmogorov complexity is that certain results are known unconditionally. In particular, in  applications there is often no need to rely on unproven conjectures from complexity theory. 
    
	\section{Unconditional Hardness of Estimating Time-Bounded Kolmogorov Complexity}\label{sec:hardness}

	In this section, we turn our attention to \emph{meta-computational} problems, which are problems that are themselves about computations and their complexity. An example of such a problem is $\mathsf{MCSP}$ (Minimum Circuit Size Problem), where we are given the truth table of a Boolean function $f \colon \{0,1\}^m \to \{0,1\}$ (represented as a Boolean string $x$ of length $n = 2^m$) and a size bound $s$,  and must decide if $f$ can be computed by a Boolean circuit containing at most $s$ gates. Similarly, we can consider the problem of computing the $\K^{t}$ complexity of an input string $x \in \{0,1\}^n$, where $t(n)$ is some fixed polynomial, such as $t(n) = n^3$. In both cases, it is not hard to see that we obtain a problem in $\mathsf{NP}$. Due to their meta-computational nature, intriguing properties (e.g.,~\citep{DBLP:conf/coco/OliveiraPS19}), and connections to other areas such as learning theory (e.g.,~\citep{DBLP:conf/coco/CarmosinoIKK16}) and cryptography (e.g.,~\citep{DBLP:conf/focs/LiuP20}), it is possible that the investigation of the complexity of meta-computational problems  can offer a fruitful path towards a proof that $\mathsf{P} \neq \mathsf{NP}$.
	
	Given the challenge of establishing strong unconditional lower bounds for problems in $\mathsf{NP}$, it is also interesting to consider the complexity of computing other notions of time-bounded Kolmogorov complexity, such as $\Kt$ and $\rKt$. For instance, given a string $x \in \{0,1\}^n$, can we efficiently estimate $\Kt(x)$? Note that this can be done in exponential time using a brute-force search, which places the decision version of this problem in $\mathsf{E} = \mathsf{DTIME}[2^{O(n)}]$. Intuitively, it seems that computing $\Kt$ and $\rKt$ should be computationally hard for the following reasons: 
	\begin{itemize}
	    \item[(\emph{i})] It looks like we must perform an exhaustive search over machines of non-trivial description length.
	    \item[(\emph{ii})]  Thanks to the definitions of $\Kt$ and $\rKt$, even the mere act of checking whether a specific machine $M$ prints the string $x$ could require an exponential time simulation. 
	\end{itemize}
	Note that (\emph{ii}) is not present in problems such as $\mathsf{MCSP}$. (We will revisit this intuition later in the section.)

	The next result shows that $\mathsf{MrKtP}$, the Minimum $\mathsf{rKt}$ Problem, is computationally hard for randomised algorithms. Indeed, even a gap version of the problem remains difficult. Note that the result provides an \emph{unconditional} complexity lower bound for a natural problem.\footnote{As observed by \cite{DBLP:conf/icalp/Oliveira19}, the problem stated next can be solved in randomised exponential time.}
	
	\begin{theorem}[Complexity Lower Bound for Estimating {$\mathsf{rKt}$}~\cite{DBLP:conf/icalp/Oliveira19}]\label{t:MrKtP-lb}
		For any $0 < \varepsilon < 1$, consider the promise problem $\Pi^{\varepsilon}_{\rKt} = (\mathcal{YES}_n, \mathcal{NO}_n)_{n \in\mathbb{N}}$, where
		\begin{eqnarray}
			\mathcal{YES}_n & = & \{x \in \{0,1\}^n \mid \mathsf{rKt} \leq n^{\varepsilon} \}, \nonumber \\
			\mathcal{NO}_n & = & \{x \in \{0,1\}^n \mid \mathsf{rKt}(x) \geq n - 1 \}. \nonumber
		\end{eqnarray}
		Then $\Pi^\varepsilon_{\rKt} \notin \mathsf{promise}$-$\mathsf{BPTIME}[n^{\polylog(n)}]$.
	\end{theorem}
	
	\begin{proof}[Proof Sketch.] The proof can be described in different ways. Here we provide a high-level exposition of the argument using insights from computational learning theory. For simplicity, we consider the weaker lower bound $\Pi^\varepsilon_{\rKt} \notin \mathsf{promise}$-$\mathsf{BPP}$.
	
	Assume towards a contradiction that $\Pi^\varepsilon_{\rKt} \in \mathsf{promise}$-$\mathsf{BPP}$. We proceed as follows.
	
	\vspace{-0.2cm}
	
	\begin{enumerate}
	    \item Under this assumption, it is possible to show that there is a (promise) \emph{natural property} (in the sense of \citep{DBLP:journals/jcss/RazborovR97}) against  functions computed by circuits of size $2^{\delta n}$, for some $\delta > 0$. In other words, we can efficiently distinguish truth-tables of bounded complexity from random truth-tables.
	    
	    \item By the main result of \citep{DBLP:conf/coco/CarmosinoIKK16}, this implies that Boolean circuits of size $s$ can be PAC learned under the uniform distribution with membership queries in time $\mathsf{poly}(s)$.
	    
	    \item Exploring the connection between learning  and circuit lower bounds from \citep{DBLP:conf/coco/OliveiraS17}, the existence of such learning algorithms implies that $\mathsf{BPE} \nsubseteq \mathsf{SIZE}(\poly)$, where $\mathsf{SIZE}(\poly)$ denotes the set of languages computed by Boolean circuits of polynomial size.
	    
	    \item Finally, we argue that if $\Pi^\varepsilon_{\rKt}$ is in $\mathsf{promise}$-$\mathsf{BPP}$ then $\mathsf{BPE} \subseteq \mathsf{SIZE}(\poly)$. Roughly speaking, this step explores techniques from pseudorandomess \citep{DBLP:journals/siamcomp/AllenderBKMR06} to show that every $L \in \mathsf{BPE}$ non-uniformly reduces to  $\Pi^\varepsilon_{\rKt}$. Since by assumption this problem can be solved by efficient probabilistic algorithms, and such algorithms can be non-uniformly simulated by polynomial-size circuits, the inclusion follows.
	\end{enumerate}
	\vspace{-0.2cm}
	
\noindent 	Given that Items 3 and 4 are in contradiction, we obtain the desired complexity lower bound. (A proof that employs a different perspective is provided in \citep{DBLP:conf/icalp/Oliveira19}.)
	\end{proof}
	
	Curiously, establishing an analogous lower bound for $\mathsf{MKtP}$ remains a notorious open problem (see, e.g.,~\citep{DBLP:journals/siamcomp/AllenderBKMR06}). Here $\mathsf{MKtP}$ refers to the problem of deciding, given a string $x$ and a positive integer $s$, whether $\Kt(x) \leq s$. While it is believed that $\mathsf{MKtP} \notin \mathsf{P}$, we currently only know how to resolve the randomised version of the problem (\Cref{t:MrKtP-lb}).\footnote{The proof of \Cref{t:MrKtP-lb} explores randomised computation to perform an indirect diagonalisation, and it is not clear how to implement a similar strategy when only deterministic computations are available.} This provides another setting where probabilistic time-bounded Kolmogorov complexity offers an advantage over its deterministic counterpart. Note that \Cref{t:MrKtP-lb} implies that $\mathsf{MKtP} \notin \mathsf{BPP}$ under a derandomisation assumption (\Cref{t:Kt-and-rKt}).

	Before presenting a different lower bound, we revise our initial intuition about the hardness of computing $\rKt$ and $\Kt$. In light of \Cref{c:short_list}, a result established after \cite{DBLP:conf/icalp/Oliveira19}, we now understand that the hardness of the gap version of $\mathsf{MrKtP}$ can be blamed on Item (\emph{ii}) only. Interestingly, an unexpected algorithmic result sheds light on the hardness of estimating $\rKt$ complexity. At the same time, this tells us that  different techniques will be needed to understand the computational hardness of problems such as $\mathsf{MCSP}$ or computing $\K^t$, where the hardness must come from the analogue of Item (\emph{i}).

	Next, we discuss a complexity lower bound for estimating the $\rK^{\poly}$ complexity of an input string.
	
		\begin{theorem}[Complexity Lower Bound for Estimating {$\mathsf{rK}^{\mathsf{poly}}$}~\cite{LOS21}]
		For any $0 < \varepsilon < 1$ and $d \geq 1$ there exists a constant $k \geq 1$ for which the following holds. Consider the  promise problem $\Pi^{\varepsilon, k}_{\rK^t} = (\mathcal{YES}_n, \mathcal{NO}_n)_{n \in\mathbb{N}}$, where
		\begin{eqnarray}
			\mathcal{YES}_n & = & \{x \in \{0,1\}^n \mid \mathsf{rK}^t(x) \leq n^{\varepsilon} \}, \nonumber \\
			\mathcal{NO}_n & = & \{x \in \{0,1\}^n \mid \mathsf{rK}^t(x) \geq n - 1 \}, \nonumber
		\end{eqnarray}
		and $t(n) = n^k$. Then $\Pi^{\varepsilon, k}_{\rK^t} \notin \mathsf{promise}$-$\mathsf{BPTIME}[n^d]$.
	\end{theorem}
	
	\begin{proof} We establish the weaker result that $\Pi^{\varepsilon, k}_{\rK^t} \notin \mathsf{promise}$-$\mathsf{DTIME}[n^d]$. The lower bound against probabilistic time can be established in a similar way, using that the pseudodeterministic PRG from \Cref{t:los-PRG} also fools probabilistic algorithms (see \citep{LOS21} for the details).
	
	Fix $0 < \varepsilon < 1$ and $d \geq 1$. Let $\varepsilon' = \varepsilon/2$, $d' = 1$, and $c' = d$. Instantiate the pseudodeterministic PRG from \Cref{t:los-PRG} with the parameters $\varepsilon'$, $c'$, and $d'$, and assume that $G_n \colon \{0,1\}^{n^{\varepsilon'}} \to \{0,1\}^n$ can be computed probabilistically in time $n^{k'}$, for some constant $k'$ (when provided with the correct advice bit $\alpha'(n)$). We let $k = 2k'$.
	
	Now suppose, towards a contradiction, that $\Pi^{\varepsilon, k}_{\rK^t} \in \mathsf{promise}$-$\mathsf{DTIME}[n^d]$. Let $A$ be a deterministic algorithm running in time $n^d$ that accepts $\mathcal{YES}_n$ and rejects $\mathcal{NO}_n$, for every large enough $n$. We argue that the existence of $A$ contradicts the infinitely often guarantee of pseudorandomness provided by the PRG $G_n$. Indeed, fix a large enough input length $n$ for which $G_n$ succeeds. On the one hand, by our choice of $k$ and $\varepsilon'$, it is easy to see that every string $y \in \{0,1\}^n$ in the image of $G_n$ satisfies $\rK^{n^k}(y) \leq n^{\varepsilon}$. For this reason, $\Pr_{z \sim \{0,1\}^{n^{\varepsilon'}}}[A(G(z)) = 1] = 1$. On the other hand, by a counting argument, a random string $x \sim \{0,1\}^n$ satisfies $\rK^{n^k}(x) \geq n-1$ with probability $\Omega(1)$ (\Cref{l:incompressible}). This implies that $\Pr_{x \sim \{0,1\}^n}[A(x) = 1] \leq 1 - \Omega(1)$, since $A$ rejects strings in $\mathcal{NO}_n$. Now notice that this violates the pseudorandomness of $G_n$. In other words, we get that $\Pi^{\varepsilon, k}_{\rK^t} \notin \mathsf{promise}$-$\mathsf{DTIME}[n^d]$.
	\end{proof}
	
		A complexity lower bound for computing $\mathsf{K}^t$ against deterministic polynomial-time algorithms and for $t = n^{\omega(1)}$ was established by Hirahara \citep{DBLP:conf/stoc/Hirahara20} using different techniques. In both cases, the time bound in the definition of the Kolmogorov complexity measure is larger than the time bound of the algorithm trying to compute or estimate Kolmogorov complexity. Needless to say, it would be extremely interesting to establish a complexity lower bound for computing Kolmogorov complexity with respect to a \emph{fixed} polynomial $t$ in $\K^t$ or $\rK^t$  that holds against \emph{arbitrary} polynomial-time algorithms (see  \citep{DBLP:conf/focs/LiuP20}). 
		
		A lower bound question that should be more accessible is presented next.

	\begin{problem}[Exponential Hardness of Estimating $\rKt$]
	Show that for any constant $0 < \varepsilon < 1$ there is a constant $\delta > 0$ such that $\Pi^{\varepsilon}_{\rKt} \notin \mathsf{promise}$-$\mathsf{BPTIME}[2^{n^{\delta}}]$.
	\end{problem}
	
	\section{Constructing Strings of Large \texorpdfstring{$\rKt$}{rKt} Complexity and Hierarchy Theorems}\label{sec:explicit_construction}

The problem of \emph{explicitly} constructing mathematical objects of different types (beyond merely showing their existence) has received much attention in computer science and mathematics. For instance, in Section \ref{sec:intro} we described the problem of deterministically producing an $n$-bit prime. In this section, we are interested in the problem of constructing \emph{incompressible strings}. Some problems of this form are particularly challenging, since given a long incompressible string (e.g., with respect to circuit size or $\K^{\poly}$ complexity), several other constructions problems can be  solved (see, e.g.,~\citep{DBLP:journals/mst/Santhanam12, DBLP:conf/focs/Korten21}).

In more detail, here we consider the problem of explicitly constructing strings that have large $\rKt$ complexity. To provide intuition, let us first consider the much simpler case of $\Kt$ complexity. Our goal is to design a deterministic algorithm that, given $1^n$, outputs an $n$-bit string $x$ such that $\Kt(x) \geq n/10$. Does this problem admit a polynomial-time algorithm? It is easy to see that this problem cannot be solved in time $2^{o(n)}$. Indeed, it follows from the very definition of $\Kt$ complexity that any deterministic algorithm $A(1^n)$ running in time $2^{o(n)}$ can only print an $n$-bit string of $\Kt$ complexity $o(n)$. However, it is not hard to see that this explicit construction problem can be solved in time $2^{O(n)}$ via an exhaustive search (for instance, by enumerating all strings produced in time $\leq 2^{n/10}$ by machines of description length $\leq n/10$). 
	
Similarly, we ask if there is an algorithm that runs in time $2^{O(n)}$ and produces an $n$-bit string $x$ such that $\rKt(x) \geq n/10$. The natural brute-force approach to solve this problems involves the simulation of \emph{randomised} algorithms. For this reason, we relax our goal as follows: Is there a \emph{randomised} algorithm $A(1^n)$ that runs in time $2^{O(n)}$ and outputs with probability at least $2/3$ a \emph{fixed} $n$-bit string $w_n$ such that $\rKt(w_n) \geq n/10$? In other words, we would like to have a \emph{pseudodeterministic} construction of strings of large $\rKt$ complexity, in the sense of \citep{DBLP:journals/eccc/GatG11}.

A careful inspection of the natural brute-force approach that works for $\Kt$ reveals that it simply does not work in the case of $\rKt$: roughly speaking, the simulation of different randomised machines comes with uncertainties, and it is not clear if after all the simulations we isolate the same string $w_n$ with high probability.
	
In \citep{LOS21}, we connected the problem of constructing strings of large $\rKt$ complexity to the longstanding question of establishing a strong time hierarchy theorem for probabilistic computations. Recall that, while it is known that $\mathsf{BPEXP} \nsubseteq \mathsf{BPP}$, it is consistent with current knowledge that inclusions such as $\mathsf{BPTIME}[2^n] \subseteq \mathsf{BPTIME}[2^{n^{0.01}}]$ and $\mathsf{BPTIME}[n^{50}] \subseteq \mathsf{BPTIME}[n^2]$ might hold.\footnote{Some separations have been established if we allow advice bits in the upper bound and lower bound. For instance, $\mathsf{BPTIME}[n^{50}]/1 \nsubseteq \mathsf{BPTIME}[n^2]/1$ (see, e.g.,~\citep{DBLP:conf/random/Barak02, DBLP:conf/focs/FortnowS04}).}

	\begin{theorem}[Explicit Construction Problem for $\rKt$ and Probabilistic Time Hierarchies] \label{t:rKt-hierarchy} The following statements are equivalent:
		\begin{itemize}
			\item[\emph{(1)}] \emph{Pseudodeterministic construction of strings of large $\rKt$ complexity:} 
			There is a constant $\varepsilon > 0$ and a randomised algorithm $A$ that, given $m$, runs in time $2^{O(m)}$ and outputs with probability at least $2/3$ a fixed $m$-bit string $w_m$ such that $\rKt(w_m) \geq \varepsilon m$.
			\item[\emph{(2)}] \emph{Strong time hierarchy theorem for probabilistic computation:} There are constants $k \geq 1$ and $\lambda > 0$ for which the following holds. For any constructive function $n \leq t(n) \leq 2^{\lambda \cdot 2^n}$, there is a language $L \in \mathsf{BPTIME}[(t(n)^k]$ such that $L \notin \mathsf{i.o.BPTIME}[t(n)]/\log t(n)$. 
					\end{itemize}
	\end{theorem}
	
	The proof of \Cref{t:rKt-hierarchy} is  elementary, and proceeds by associating with a language $L$ a sequence of truth-tables, one for each input length $n$ (each truth-table can be seen as a string of length $m = 2^n$). For a sketch of the argument and a detailed proof, see \citep{LOS21}.
	
	Note that the connection between the explicit (pseudodeterministic) construction problem for $\rKt$ and hierarchy theorems goes in both ways.	More generally, \citep{LOS21} explored  the fruitful relation between pseudodeterministic PRGs (see \Cref{sec:primes_PRGs}), the explicit construction problem for $\rKt$ complexity, and  hierarchy theorems for probabilistic time to make advances in all these areas. On the other hand, \citep{DBLP:conf/icalp/LuO21} connected $\rKt$ complexity and its coding theorem (\Cref{t:rkt-coding}) to the study of time hierarchy theorems for sampling distributions (cf.,~\citep{DBLP:journals/siamcomp/Watson14}).

	\section{Concluding Remarks}\label{sec:final}

	We presented key results in probabilistic Kolmogorov complexity and applications to several areas, including explicit constructions, complexity lower bounds, sampling algorithms, average-case complexity, and learning theory. The probabilistic measures $\rK^t$, $\pK^t$, and $\rKt$ are particularly useful in settings that involve randomised algorithms. While it is quite possible for these complexity measures  to be essentially equivalent to their deterministic counterparts (\Cref{sec:prob_vs_det}), they allow us to obtain unconditional results that do not rely on derandomisation assumptions. In some cases, probabilistic Kolmogorov complexity can significantly simplify existing arguments or is the only known approach to certain results.
	
	The results presented in the preceding sections naturally suggest several problems and directions. For example, we believe that it should be possible to make progress on the following fronts:
	\begin{enumerate}
	    \item[--] Designing improved pseudodeterministic PRGs and obtaining better upper bounds on the $\rKt$ complexity of prime numbers.
	    \item[--] Establishing new unconditional lower bounds on the complexity of meta-computational problems such as $\mathsf{MKtP}$ and $\mathsf{MrKtP}$.
	\end{enumerate}
	For a more precise formulation of these problems, we refer to the concrete questions stated in the corresponding sections of the article (\Cref{sec:primes_PRGs} and \Cref{sec:hardness}). Additional questions of interest are presented in other parts of the survey.
	
	Given the number of recent advances and applications of time-bounded Kolmogorov complexity to algorithms and complexity theory (see \Cref{sec:intro}), it is hard to predict which directions will be more fruitful. Nevertheless, we are particularly optimistic about the role that probabilistic Kolmogorov complexity can take in the investigation of the relations between average-case complexity and worst-case complexity, cryptography, and learning algorithms. In particular, analogously to results of \citep{DBLP:conf/stoc/Hirahara21},  under the assumption that $\mathsf{DistNP} \subseteq \mathsf{AvgBPP}$, all main pillars of Kolmogorov complexity are known to hold for $\pK^t$ complexity: incompressibility (\Cref{l:incompressible}), coding theorem (\Cref{t:pkt_coding}), language compression (\Cref{t:language_compression}), and symmetry of information (\Cref{t:SoI}). Taking into account the wide applicability of these results and the ubiquitous role of randomised algorithms in theoretical computer science, we expect to see further developments in average-case complexity powered by tools and perspectives from probabilistic Kolmogorov complexity.

	\noindent

	\bibliographystyle{alpha}	
	\bibliography{references}

\end{document}